\documentclass[journal]{IEEEtran}
\usepackage[T1]{fontenc}
\usepackage{bm}
\usepackage{amsthm}
\usepackage{amsmath}
\usepackage{amssymb}
\usepackage{graphicx}
\usepackage{subfig}
\usepackage{xcolor}
\usepackage{color, colortbl}
\definecolor{Gray}{gray}{0.9}
\usepackage[first=0,last=9]{lcg}
\makeatletter

\providecommand{\tabularnewline}{\\}

\theoremstyle{plain}
\newtheorem{thm}{\theoremname}
\theoremstyle{definition}
\newtheorem{defn}{\definitionname}
\theoremstyle{remark}

\theoremstyle{plain}
\newtheorem{lem}{\lemmaname}
\theoremstyle{plain}
\newtheorem{prop}{\propositionname}

\DeclareMathOperator*{\argmin}{arg\,min}

\newcommand{\norm}{\text{Norm}}

\usepackage[vlined,boxed,ruled,linesnumbered,resetcount]{algorithm2e}
\usepackage{algpseudocode}

\usepackage[noadjust,sort,compress]{cite}
\makeatother

\providecommand{\definitionname}{Definition}
\providecommand{\lemmaname}{Lemma}
\providecommand{\propositionname}{Proposition}
\providecommand{\remarkname}{Remark}
\providecommand{\theoremname}{Theorem}
\date{}

\begin{document}

\title{Lattice Reduction over Imaginary Quadratic Fields}

\author{Shanxiang Lyu,~\IEEEmembership{Member,~IEEE}, Christian Porter and Cong Ling,~\IEEEmembership{Member,~IEEE}
	\thanks{The
		work of S. Lyu was supported in part by the National Natural Science Foundation of China under Grants 61902149 and 62032009, in part by the Natural Science Foundation of Guangdong Province under Grant 2020A1515010393, in part by the Fundamental Research Funds for the Central Universities under Grant 21620438, and in part by the Major Program of Guangdong Basic and Applied Research under Grant 2019B030302008. 
		The material in this paper was presented in part at the IEEE Information Theory Workshop 2018, Guangzhou, China, and the IEEE Information Theory Workshop 2019, Visby, Gotland, Sweden.} 
	\thanks{S. Lyu is with College of Cyber Security,
		Jinan University,
		Guangzhou 510632, China (e-mail: shanxianglyu@gmail.com), and also with the State Key Laboratory of Cryptology, P.O. Box 5159, Beijing, 100878, China. }
	\thanks{ C. Porter and C. Ling are with the Department of Electrical
		and Electronic Engineering, Imperial College London, London SW7 2AZ,
		United Kingdom (e-mail: c.porter17@imperial.ac.uk, cling@ieee.org). }}

\maketitle
\begin{abstract}
Complex bases, along with direct-sums defined by rings of imaginary quadratic
integers, induce algebraic lattices. In this work, we study such lattices and their reduction algorithms. 
Firstly, when the lattice is spanned over a two dimensional basis, we show that the algebraic variant of Gauss's algorithm returns a basis that corresponds to the successive minima of the lattice if the chosen ring is Euclidean. Secondly, we extend the celebrated Lenstra-Lenstra-Lov\'asz (LLL) reduction from over real bases to over complex bases. Properties and implementations of the algorithm are examined. In particular, satisfying Lov\'asz's condition requires the ring to be Euclidean. Lastly, we numerically show the time-advantage of using algebraic LLL by considering lattice bases generated from wireless communications and cryptography.
\end{abstract}


\section{Introduction}

\IEEEPARstart{L}{attice}  reduction is to find a basis with short and nearly orthogonal vectors when given a basis as input. Its applications in signal processing, information theory, and cryptology include: designing finite wordlength FIR filters \cite{lion18},
reducing the channel matrices in lattice-reduction-aided MIMO detection/precoding \cite{Wuebben2011,park11,park12,ahmad13,wkma14,Liu2012a}, designing the network coding coefficients in compute-and-forward \cite{FSK13,Sun2013,yuan16}, 
cryptanalysing lattice-based cryptographic systems \cite{Nguyen2010}, etc.
Initially the lattices involved feature direct-sums defined over integers $\mathbb{Z}$ or Gaussian integers $\mathbb{Z}[i]$ \cite{Wuebben2011,FSK13}, but in recent years there has been a surge on using more compact signal constellations and algebraic lattice codes \cite{Tunali2015,jerry2018,Stern2016}. 
Thanks to the algebraic structure, reduction algorithms utilizing this advantage \cite{Gan2009,Stern2016} often save a large amount of computational cost.

On the one hand,
lattice reduction has been well investigated for conventional $\mathbb{Z}$-lattices. Some of the reduction algorithms include: 
the celebrated Lenstra-Lenstra-Lov\'asz (LLL)  \cite{Lenstra1982} and its variants  \cite{Schnorr1994,Chang2005,Lyu2017}, block-Korkine-Zolotarev (BKZ)  \cite{Chen2011b}, Korkine-Zolotarev (KZ)  \cite{Lagarias1990,wen19}, and Minkowski 
\cite{Minkowski1905,Zhang2012tsp}. On the other hand,
for $\mathcal{O}_{\mathbb{K}}$-lattices where $\mathcal{O}_{\mathbb{K}}$
denotes the ring of integers of a number field $\mathbb{K}$, the
reduction techniques can be classified based on whether the lattice
vectors lie in $\mathcal{O}_{\mathbb{K}}$ or the complex field $\mathbb{C}$.
The first scenario arises quite often in lattice-based cryptography,
and much work has been done in generalizing LLL for such lattices
\cite{Napias1996,Fieker1996,Fieker2010,Kim2017}. Napias's work \cite{Napias1996} extends LLL to lattices defined by Euclidean rings contained in a CM number field or a quaternion field.  Fieker and Pohst's approach \cite{Fieker1996}
defines LLL over Dedekind domains, while Fieker and Stehl\'e's approach
\cite{Fieker2010} is to apply LLL to an equivalent higher dimensional
$\mathbb{Z}$-lattice and return this to a module. Quite recently,
Kim and Lee \cite{Kim2017} presented reduction algorithms for arbitrary
Euclidean domains. Regarding the second scenario whose basis vectors
are in $\mathbb{C}$, the LLL algorithm has also been generalized
to $\mathbb{Z}[i]$-lattices \cite{Gan2009},  $\mathbb{Z}[\omega]$-lattices
\cite{Stern2016}, and lattices from imaginary quadratic fields \cite{japan1,japan2}. 

As a motivation, we notice that a general study on
the reduction of $\mathbb{Z}[\xi]$-lattices, where $\mathbb{Z}[\xi]$
denotes a ring of imaginary quadratic integers, is lacking. Questions that remain unanswered within such lattice reduction include: \textit{What are the behaviors of the fundamental lattice parameters (e.g., Hermite's constant and Minkowski's theorems)? 
Can Gauss's algorithm output successive minima when the two-dimensional lattices are algebraic? Can we properly define algebraic LLL for all types of rings? Which type of rings leads to shorter lattice vectors? How much faster the algebraic algorithms can be?}

In this work, we seek to better understand the characteristics of algebraic
lattices, along with the proper design and performance limits of reduction algorithms. The contributions
of this paper are the following:

i) After presenting the definitions and measures for algebraic lattices,
we analyze the algebraic analogs for the orthogonality defect, Hermite's
factor, and Minkowski's first and second theorems. Furthermore, we extend the definition of lattice reduction from over $\mathbb{Z}$-lattices
to $\mathbb{Z}[\xi]$-lattices, which says that the reduction is to
find a unimodular matrix from the general linear group $\mathrm{GL}_{n}\left(\mathbb{Z}\left[\xi\right]\right)$. The relation between unimodular matrices and independence of linear codes over finite fields is analyzed.

ii) For lattices
of two dimensions, we take a modest step to investigate algebraic Gauss's
algorithm. When the ring of integers
is a Euclidean domain, we prove that Gauss's algorithm returns a basis
corresponding to the successive minima of an algebraic lattice. This
result is further explained through numerical examples. Specifically,
we show how the algorithm finds the two successive minima when the
domain is Euclidean, and how the algorithm fails to work when it is
non-Euclidean.

iii) For higher dimensional lattices, we investigate the algebraic version of the celebrated LLL algorithm. By
analyzing Siegel's condition and the covering radiuses of rings in complex quadratic fields, 
the lower bound of Lov\'asz's parameter $\delta$ is derived. To ensure the algorithm is convergent, we show the rings have to be Euclidean.
Although we can always transform a
$\mathbb{Z}[\xi]$-lattice to a $\mathbb{Z}$-lattice and perform conventional LLL reduction, algebraic LLL algorithms are less time-consuming, and we also have a better reduction factor when rounding over the Eisenstein integers. Moreover, conventional LLL additionally needs a time consuming process to return the $2n$ real vectors to $n$ linear independent vectors (if the application requires).
 
iv) We numerically verify the efficiency of algebraic LLL via applications to both wireless communications and cryptanalysis. Reduction is performed over bases from  compute-and-forward  \cite{FSK13,Sun2013,Tunali2015,jerry2018,yuan16}, lattice-reduction-aided \cite{park11,park12,ahmad13,wkma14} and integer-forcing equalization \cite{Zhan2014IT,Fischer19}, and extended NTRU cryptosystems \cite{GNTRU06,ETRU15}. The algebraic algorithms save at least $50\%$ of the running time. In addition, we demonstrate how to cryptanalyze the celebrated NTRU crypto system \cite{HoffsteinPS98,StehleS11} based on quadratic subfields of cyclotomic fields. 
 
The rest of this paper is organized as follows. In  Section \ref{sec:PreliminariesITW}, backgrounds about 
quadratic fields and algebraic lattices are reviewed, and the concept of algebraic lattice reduction is induced.
The definitions and properties of algebraic Gauss reduction and algebraic LLL reduction  are presented
in Sections \ref{alGaussSection} and \ref{sec:LLL-over-modules}, respectively. In Section \ref{sec:app-lr-cf},
we present numerical results.  Concluding remarks are given in the last section.

Notations: Matrices and column vectors are denoted by uppercase and lowercase
boldface letters, respectively.   The real and imaginary parts
of a complex number are denoted as $\mathfrak{R}\left(\cdot\right)$ and $\mathfrak{I}\left(\cdot\right)$.
$\mathbb{Z}$, $\mathbb{Z}\left[i\right]$, $\mathbb{\mathbb{Z}\left[\omega\right]}$, $\mathbb{Q}$, $\mathbb{R}$ and $\mathbb{C}$ are used to denote
the set of integers, Gaussian integers, Eisenstein integers, rational, real, and complex numbers, respectively.  $\mathbb{F}_{p}$ denotes a finite field of size $p$.  $\left(\cdot\right)^{\dagger}$ refers to the conjugate (transpose) of either a scalar or a matrix. $|\cdot|^{2}$ and $\left\Vert \cdot\right\Vert ^{2}$
respectively denote Euclidean norm of a scalar and a vector. $\otimes$ refers to Kronecker tensor product.
 $V_{n}$
refers to the volume of a unit ball in $\mathbb{R}^{n}$.  $\mathbb{E}$ denotes expectation.

\section{\label{sec:PreliminariesITW}Algebraic Lattices and Reduction}

Lattice reduction is closely related to number theory. In this subsection,
we review some concepts  that will be used throughout
this paper. We refer readers to \cite{BK:Mollin-ANT,oggier04tutorial} for a more detailed account of algebraic concepts, e.g., the definitions of groups, rings, and fields.

\begin{defn}[Number field]
In mathematics, an algebraic number field (or simply number field) $\mathbb{K}$ is a finite degree  field extension of the field of rational numbers $\mathbb{Q}$.
\end{defn}

\begin{defn}[Quadratic field]
A quadratic field is an algebraic number field $\mathbb{K}$ of degree
$[\mathbb{K}:\mathbb{Q}]=2$ over $\mathbb{Q}$. In particular, we
write $\mathbb{K}=\mathbb{Q}\left(\sqrt{-d}\right)$ where $d\in\mathbb{Z}$
is square free. If $d>0$, we say $\mathbb{Q}\left(\sqrt{-d}\right)$
is an imaginary quadratic field.
\end{defn}

\begin{defn}[Algebraic integer]
An algebraic integer is a complex number which is a root of some
monic polynomial whose coefficients are in $\mathbb{Z}$, where a monic polynomial is a single-variable polynomial in which the nonzero coefficient of highest degree is equal to $1$. 
\end{defn}

The set
of all algebraic integers forms a subring $\mathcal{S}$ of $\mathbb{C}$.
For any number field $\mathbb{K}$, we write $\mathcal{O}_{\mathbb{K}}=\mathbb{K}\cap\mathcal{S}$
and call $\mathcal{O}_{\mathbb{K}}$ the ring of integers of $\mathbb{K}$. The fact that the algebraic integers of $\mathbb{K}$ form a ring is a strong result \cite{oggier04tutorial}.
Regarding the ring of integers of a quadratic field $\mathbb{Q}\left(\sqrt{-d}\right)$,
one has $\mathcal{O}_{\mathbb{K}}=\mathbb{Z}\left[\xi\right]$ where
\[
\xi=\begin{cases}
\sqrt{-d}, & \mathrm{if}\thinspace-d\equiv2,3\thinspace\mod\thinspace4; \mathrm{(Type\,\,I)} \\
\left(1+\sqrt{-d}\right)/2 & \mathrm{if}\thinspace-d\equiv1\thinspace\mod\thinspace4. \mathrm{(Type\,\,II)}
\end{cases}
\]
The two types of rings are respectively referred to as Type I and Type II. E.g., Gaussian integers $\mathbb{Z}\left[i\right] \triangleq \mathbb{Z}\left[\sqrt{-1}\right]$ is a Type I ring, and Eisenstein integers  $\mathbb{\mathbb{Z}\left[\omega\right]}\triangleq \mathbb{Z}\left[{(1+\sqrt{-3})/2}\right]$ is a Type II ring.

\begin{defn}[Euclidean domain] \label{def_euclidean}
	A Euclidean domain is an integral domain which can be endowed with at least one Euclidean function.
	For the ring  $\mathbb{Z}\left[\xi\right]$, a Euclidean function $\phi$ is a map from $\mathbb{Z}\left[\xi\right] \setminus 0 $ to the non-negative integers such that $\phi \left(a\right)\leq \phi\left(ab\right)$ for any nonzero $a,b\in \mathbb{Z}\left[\xi\right]$, and there exist $q$ and $r$ in $\mathbb{Z}\left[\xi\right]$ such that $a=bq+r$ with $r=0$ or $\phi\left(r\right)<\phi\left(b\right)$.
\end{defn}

\begin{defn}[Algebraic norm] 
The algebraic norm function $\mathrm{Nr}(\cdot)$ of an element is $\mathrm{Nr}\left(a+b\xi\right)=\left(a+b\xi\right)\left(a+b\xi\right)^{\dagger}$. 
\end{defn}

Elements
in $\mathbb{Z}\left[\xi\right]$ with norm $\pm1$ are called the
units of $\mathbb{Z}\left[\xi\right]$. Together they form a unit
group denoted by $\mathbb{Z}\left[\xi\right]^{\times}$. E.g., the 
unit group of $\mathbb{Z}\left[i\right]$ is $\mathbb{Z}\left[i\right]^{\times}=\{0,1,i,-i\}$.
If the Euclidean function $\phi$ is defined by the algebraic norm, then the Euclidean domain is called norm-Euclidean \cite{Kim2017}. \textit{For complex quadratic fields, Euclidean norm coincides with algebraic
norm, i.e.,} $\mathrm{Nr}\left(a+b\xi\right)=|a+b\xi|^{2}$. Thus a Euclidean ring of imaginary quadratic integers is also norm-Euclidean.

For imaginary
quadratic fields, we may analytically extend the norm function to
all complex numbers using the absolute value. Moreover, we have $\max_{x\in K}\min_{q\in\mathcal{O}_{K}}|\norm_{K/\mathbb{Q}}(x-q)|=\max_{x\in\mathbb{C}}\min_{q\in\mathcal{O}_{K}}|x-q|^{2}$
as the maximum distance with respect to the absolute value is achieved
at a rational point. 

\begin{defn}[Module]
A $\mathbb{Z}\left[\xi\right]$-module is a set $M$ together with   a binary operation under
which $M$ forms an Abelian group, and an action of $\mathbb{Z}\left[\xi\right]$ on $M$ which 
satisfies the same axioms as those for vector spaces.
\end{defn}

A module may not have a basis (i.e., not free). A subset of $M$ forms a $\mathbb{Z}\left[\xi\right]$-module
basis of $M$ if elements of subset are linearly independent over $\mathbb{Z}\left[\xi\right]$ and if every element in $M$ can be written as a finite
linear combination of elements in the subset.
In this work, we will call any free $\mathbb{Z}\left[\xi\right]$-module an algebraic lattice.

\begin{defn}[Algebraic lattice\footnote{Since $\mathbf{b}_{1},\ldots\thinspace,\mathbf{b}_{n}$ have rank $n$ over  $\mathbb{C}$, they are also linearly independent over $\mathbb{Z}\left[\xi\right]$. The considered $\mathbb{Z}\left[\xi\right]$-modules always have bases, so $\mathbb{Z}\left[\xi\right]$ is not confined to be a principle ideal domain when defining algebraic lattices.}]
A  $\mathbb{Z}\left[\xi\right]$-lattice is a discrete
$\mathbb{Z}\left[\xi\right]$-submodule of $\mathbb{C}^{n}$ that has a basis. Such a rank $n$ lattice $\Lambda^{\mathbb{Z}\left[\xi\right]}(\mathbf{B})$ with basis
  $\mathbf{B}=[\mathbf{b}_{1},\ldots\thinspace,\mathbf{b}_{n}]\in\mathbb{C}^{n\times n}$
can be represented by 
\[
\Lambda^{\mathbb{Z}\left[\xi\right]}(\mathbf{B})=\mathbb{Z}\left[\xi\right]\mathbf{b}_{1}+\mathbb{Z}\left[\xi\right]\mathbf{b}_{2}+\cdots+\mathbb{Z}\left[\xi\right]\mathbf{b}_{n}.
\]
\end{defn}

\begin{defn}[Successive minimum]
\label{def:smpok}The $j$th successive minimum of a $\mathbb{Z}\left[\xi\right]$-lattice
$\Lambda^{\mathbb{Z}\left[\xi\right]}$ is the smallest real number
$r$ such that its embedded $\mathbb{Z}$-lattice through a bijection
$\sigma$ contains $j$ linearly $\mathbb{Z}\left[\xi\right]$-independent
vectors of length at most $r$: $\lambda_{j,\mathbb{Z}\left[\xi\right]}=$
\[
\inf\left\{ r\mathrel{\Big|}\dim\left(\mathrm{span}\left(\sigma^{-1}\left(\sigma\left(\Lambda^{\mathbb{Z}\left[\xi\right]}\right)\cap\mathcal{B}(\mathbf{0},r)\right)\right)\right)\geq j\right\} ,
\]
where $\mathcal{B}(\mathbf{t},r)$ denotes a ball centered at $\mathbf{t}$
with radius $r$, and the $r$ is measured by the complex Euclidean norm.
\end{defn}

The successive minima are hard to compute/approximate for non-algebraic random lattices \cite{Micciancio2002}, and we conjecture that this hardness still looms for algebraic lattices.
The problem of fining the first successive minimum is called the shortest vector problem (SVP), in which lattice reduction is a popular approach to obtain an approximate solution. The definition of SVP in this paper is:
\begin{defn}[SVP]
	Given a  $\mathbb{Z}\left[\xi\right]$-lattice $\Lambda^{\mathbb{Z}\left[\xi\right]}(\mathbf{B})$, 
	find a vector $\mathbf{v}\in \Lambda^{\mathbb{Z}\left[\xi\right]}(\mathbf{B})$, $\mathbf{v}\neq \mathbf{0}$ such that $\|\mathbf{v}\|\leq \|\mathbf{w}\|$ $\forall\, \mathbf{w}\in \Lambda^{\mathbb{Z}\left[\xi\right]}(\mathbf{B})$, $\mathbf{w}\neq \mathbf{0}$.
\end{defn}

\subsection{Hermite's Constant and Orthogonality Defect}

To proceed, we first show the $\mathbb{Z}$-basis (real generator matrix) of lattice $\Lambda^{\mathbb{Z}\left[\xi\right]}\left(\mathbf{B}\right)$
is: 
\begin{align}
\mathbf{B}^{\mathbb{R},\mathbb{Z}\left[\xi\right]}= & \begin{cases}
\left[\begin{array}{cc}
\mathfrak{R}\left(\mathbf{B}\right) & -\sqrt{d}\mathfrak{I}\left(\mathbf{B}\right)\\
\mathfrak{I}\left(\mathbf{B}\right) & \sqrt{d}\mathfrak{R}\left(\mathbf{B}\right)
\end{array}\right]\thinspace\thinspace\thinspace\thinspace\thinspace\mathrm{if}\thinspace\xi=\sqrt{-d};\\
\left[\begin{array}{cc}
\mathfrak{R}\left(\mathbf{B}\right) & \frac{1}{2}\mathfrak{R}\left(\mathbf{B}\right)-\frac{\sqrt{d}}{2}\mathfrak{I}\left(\mathbf{B}\right)\\
\mathfrak{I}\left(\mathbf{B}\right) & \frac{1}{2}\mathfrak{I}\left(\mathbf{B}\right)+\frac{\sqrt{d}}{2}\mathfrak{R}\left(\mathbf{B}\right)
\end{array}\right]\thinspace\thinspace\mathrm{if}\thinspace\xi=\frac{1+\sqrt{-d}}{2}.
\end{cases}\label{eq:geMatrix}
\end{align}
Denote the coefficient of a lattice vector $\mathbf{B}\mathbf{x}$ as  $\mathbf{x}=\mathbf{x}_{a}+\xi\mathbf{x}_{b}\in\mathbb{Z}\left[\xi\right]^{n}$.
If $\xi=\sqrt{-d},\thinspace d>0$, we have 
\begin{align}
&\left(\mathfrak{R}\left(\mathbf{B}\right)+i\mathfrak{I}\left(\mathbf{B}\right)\right)\left(\mathbf{x}_{a}+i\sqrt{d}\mathbf{x}_{b}\right) \nonumber\\
&=\left(\mathfrak{R}\left(\mathbf{B}\right)\mathbf{x}_{a}-\sqrt{d}\mathfrak{I}\left(\mathbf{B}\right)\mathbf{x}_{b}\right)+i\left(\mathfrak{I}\left(\mathbf{B}\right)\mathbf{x}_{a}+\sqrt{d}\mathfrak{R}\left(\mathbf{B}\right)\mathbf{x}_{b}\right);\label{eq:giMatrix}
\end{align}
and if $\xi=\frac{1+\sqrt{-d}}{2},\thinspace d>0$, we have 
\begin{align}
  & \left(\mathfrak{R}\left(\mathbf{B}\right)+i\mathfrak{I}\left(\mathbf{B}\right)\right)\left(\mathbf{x}_{a}+\frac{1}{2}\mathbf{x}_{b}+i\frac{\sqrt{d}}{2}\mathbf{x}_{b}\right)\nonumber \\
 & =\left(\mathfrak{R}\left(\mathbf{B}\right)\mathbf{x}_{a}+\left(\frac{1}{2}\mathfrak{R}\left(\mathbf{B}\right)-\frac{\sqrt{d}}{2}\mathfrak{I}\left(\mathbf{B}\right)\right)\mathbf{x}_{b}\right) \nonumber \\ &+i\left(\mathfrak{I}\left(\mathbf{B}\right)\mathbf{x}_{a}+\left(\frac{1}{2}\mathfrak{I}\left(\mathbf{B}\right)+\frac{\sqrt{d}}{2}\mathfrak{R}\left(\mathbf{B}\right)\right)\mathbf{x}_{b}\right).\label{eq:gwmatrix}
\end{align}
Define the function $\Psi:\mathbb{C}^{n}\rightarrow\mathbb{R}^{2n}$ that maps
the complex vector $\left[v_{1},\ldots\thinspace,v_{n}\right]^{\top}$
to the real vector: 
\[
\left[\mathfrak{R}\left(v_{1}\right),\ldots\thinspace,\mathfrak{R}\left(v_{n}\right),\mathfrak{I}\left(v_{1}\right),\ldots\thinspace,\mathfrak{I}\left(v_{1}\right)\right]^{\top}.
\]
By applying the mapping function $\Psi\left(\cdot\right)$ to Eqs.
(\ref{eq:giMatrix}) and (\ref{eq:gwmatrix}) for a lattice point $\mathbf{B}\mathbf{x}\in\Lambda^{\mathbb{Z}\left[\xi\right]}$, we obtain $\Psi\left(\mathbf{B}\mathbf{x}\right)=\mathbf{B}^{\mathbb{R},\mathbb{Z}\left[\xi\right]}\left[\mathbf{x}_{a},\mathbf{x}_{b}\right]^{\top}$,
where the expression for $\mathbf{B}^{\mathbb{R},\mathbb{Z}\left[\xi\right]}$
is given in (\ref{eq:geMatrix}).

According to Eq. (\ref{eq:geMatrix}), the generator matrix of $\Lambda^{\mathbb{Z}\left[\xi\right]}\left(\mathbf{B}\right)$
is related to that of the $\mathbb{Z}\left[i\right]$-lattice $\Lambda^{\mathbb{Z}\left[i\right]}\left(\mathbf{B}\right)$
via 
\begin{equation}
\mathbf{B}^{\mathbb{R},\mathbb{Z}\left[\xi\right]}=\mathbf{B}^{\mathbb{R},\mathbb{Z}\left[i\right]}\left(\Phi^{\mathbb{Z}\left[\xi\right]}\otimes\mathbf{I}_{n}\right),\label{eq:zizlRelation}
\end{equation}
where 
\begin{align}
\Phi^{\mathbb{Z}\left[\xi\right]} \triangleq & \begin{cases}
\left[\begin{array}{cc}
1 & 0\\
0 & \sqrt{d}
\end{array}\right]\thinspace\thinspace\thinspace\thinspace\thinspace\mathrm{if}\thinspace\xi=\sqrt{-d};\\
\left[\begin{array}{cc}
1 & \frac{1}{2}\\
0 & \frac{\sqrt{d}}{2}
\end{array}\right]\thinspace\thinspace\mathrm{if}\thinspace\xi=\frac{1+\sqrt{-d}}{2};
\end{cases}\label{eq:geMatrixPhi}
\end{align}
is referred to as the generator
matrix of $\mathbb{Z}\left[\xi\right]$ in $\mathbb{R}^{2}$. It follows
that we can define the volume of an algebraic lattice as 
\begin{align}
\mathrm{Vol}\left(\Lambda^{\mathbb{Z}\left[\xi\right]}\right) &\triangleq\mathrm{Vol}\left(\Lambda^{\mathbb{Z}}\left(\mathbf{B}^{\mathbb{R},\mathbb{Z}\left[\xi\right]}\right)\right)\nonumber \\&=\det\left(\mathbf{B}^{\mathbb{R},\mathbb{Z}\left[i\right]}\right)\det\left(\Phi^{\mathbb{Z}\left[\xi\right]}\right)^{n},\label{eq:volume}
\end{align}
where $\det\left(\mathbf{B}^{\mathbb{R},\mathbb{Z}\left[i\right]}\right)=|\det\left(\mathbf{B}\right)|^{2}$
denotes the volume of lattice $\Lambda^{\mathbb{Z}}\left(\mathbf{B}^{\mathbb{R},\mathbb{Z}\left[i\right]}\right)$. 

With the definition of volumes,  we extend the definition of Hermite's constant to an analogous constant for algebraic lattices. Previously, the supremum of $\lambda_{1}^{2}(\Lambda^{\mathbb{Z}})/\mathrm{Vol}\left(\Lambda^{\mathbb{Z}}\right)^{2/n}$
for all rank $n$ $\mathbb{Z}$-lattices $\Lambda^{\mathbb{Z}}$ is
often denoted by $\gamma_{n}$ and called \textit{Hermite's constant} \cite{Nguyen2010}. 
\begin{defn}[Algebraic Hermite's constant, \cite{Nguyen2010}]
\label{def:hermite} We denote by $\gamma_{n}^{\mathbb{Z}\left[\xi\right]}$
and call the supremum of ${\lambda_{1}^{2}\left(\Lambda^{\mathbb{Z}\left[\xi\right]}\right)}/{\mathrm{Vol}\left(\Lambda^{\mathbb{Z}\left[\xi\right]}\right)^{1/n}}$
for all rank $n$ $\mathbb{Z}[\xi]$-lattices $\Lambda^{\mathbb{Z}\left[\xi\right]}$
algebraic Hermite's constant. 
\end{defn}

Obviously an algebraic lattice $\Lambda^{\mathbb{Z}\left[\xi\right]}\left(\mathbf{B}\right)$
of dimension $n$ can always be described by a real lattice $\Lambda^{\mathbb{Z}}\left(\mathbf{B}^{\mathbb{R},\mathbb{Z}\left[\xi\right]}\right)$
of dimension $2n$. Moreover, since 
\begin{equation}
{\lambda_{1}^{2}\left(\Lambda^{\mathbb{Z}\left[\xi\right]}\right)}/{\mathrm{Vol}\left(\Lambda^{\mathbb{Z}\left[\xi\right]}\right)^{1/n}}\leq\gamma_{2n},\label{eq:boundGamma}
\end{equation}
we arrive at the following result: 
\[
\gamma_{n}^{\mathbb{Z}\left[\xi\right]} \leq \gamma_{2n}\leq4\left(V_{2n}^{-1/n}\right)
\]
for all positive integers $n$, in which the last inequality is from
\cite{Nguyen2010}. This upper bound
behaves independently of the chosen ring $\mathbb{Z}\left[\xi\right]$. The actual \textit{Hermite's factor},
${\lambda_{1}^{2}\left(\Lambda^{\mathbb{Z}\left[\xi\right]}\right)}/{\mathrm{Vol}\left(\Lambda^{\mathbb{Z}\left[\xi\right]}\right)^{1/n}}$,
however depends on   $\mathbb{Z}\left[\xi\right]$. 
In Fig.
\ref{fig1_lambda1}, we plot the empirical cumulative distribution
functions (CDFs) of a 2-D $\mathbb{Z}\left[\xi\right]$-lattice, where
entries of the complex basis generated from a complex Gaussian distribution
$\mathcal{CN}(0,1)$. It is known that $\gamma_{4}=\sqrt{2}$ \cite{Nguyen2010},
so this serves as the upper bound in the plot. The figure shows that
a ring with a smaller $\det\left(\Phi^{\mathbb{Z}\left[\xi\right]}\right)$
has a larger Hermite's factor on the average.
\begin{figure}
	\center
	
	\includegraphics[clip,width=0.4\textwidth]{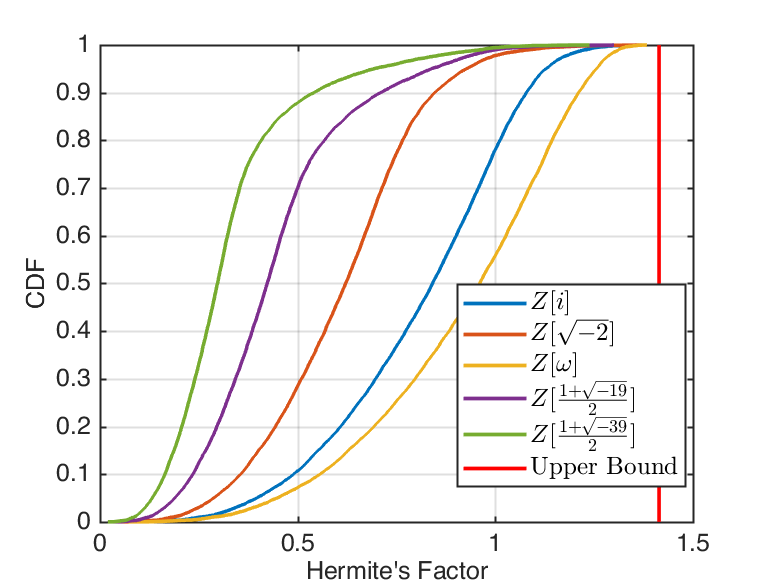}
	
	\protect\caption{The empirical cumulative distribution functions of Hermite's factor
		${\lambda_{1}^{2}\left(\Lambda^{\mathbb{Z}\left[\xi\right]}\right)}/{\mathrm{Vol}\left(\Lambda^{\mathbb{Z}\left[\xi\right]}\right)^{1/n}}$
		in 2-D lattices. }
	\label{fig1_lambda1} 
\end{figure}

Similarly, we introduce the \textit{orthogonality defect}  for algebraic
lattices: 
\begin{equation}
\eta_{\mathbb{Z}\left[\xi\right]}(\mathbf{B})\triangleq\frac{\prod_{j=1}^{n}\left\Vert \mathbf{b}_{j}\right\Vert }{\mathrm{Vol}\left(\Lambda^{\mathbb{Z}\left[\xi\right]}\right)},\label{eq:OD}
\end{equation}
which quantifies how close the basis is to being ``orthogonal''.
For a $\mathbb{Z}\left[i\right]$-lattice, its lower bound is $\eta_{\mathbb{Z}\left[i\right]}(\mathbf{B})\geq1$
according to Hadamard's inequality. More generally, it follows from
Eq. (\ref{eq:volume}) that 
\[
\eta_{\mathbb{Z}\left[\xi\right]}(\mathbf{B})\geq\det\left(\Phi^{\mathbb{Z}\left[\xi\right]}\right)^{-n}.
\]
The volume of a lattice is fixed, so the smallest $\eta_{\mathbb{Z}\left[\xi\right]}(\mathbf{B})$
is achieved only when each $\left\Vert \mathbf{b}_{j}\right\Vert $
is minimized.

\subsection{Minkowski's Theorems}

Minkowski's first and second theorems are crucial for analyzing the
performance of a lattice reduction algorithm. These theorems over
$\mathbb{Z}$-lattices are well known. For algebraic lattices where
the bases may not belong to a number field, we need the following
theorem: 
\begin{thm}[{{Minkowski's first and second theorems over $\mathbb{Z}\left[\xi\right]$-lattices}}]
\label{thm:MinThms}For a $\mathbb{Z}\left[\xi\right]$-lattice $\Lambda^{\mathbb{Z}\left[\xi\right]}\left(\mathbf{B}\right)$
with basis $\mathbf{B}\in\mathbb{C}^{n\times n}$, it satisfies 
\begin{equation}
\lambda_{1,\mathbb{Z}\left[\xi\right]}^{2}\leq\gamma_{2n}\left|\det\left(\Phi^{\mathbb{Z}\left[\xi\right]}\right)\right||\det\left(\mathbf{B}\right)|^{2/n};\label{eq:MFTmo}
\end{equation}
\begin{equation}
\prod_{j=1}^{n}\lambda_{j,\mathbb{Z}\left[\xi\right]}^{2}\leq\gamma_{2n}^{n}\left|\det\left(\Phi^{\mathbb{Z}\left[\xi\right]}\right)\right|^{n}|\det\left(\mathbf{B}\right)|^{2}.\label{eq:MSTmo}
\end{equation}
\end{thm}
\begin{IEEEproof}
Minkowski's first theorem is a direct consequence of (\ref{eq:boundGamma}).
To obtain Minkowski's second theorem for $\mathbb{Z}\left[\xi\right]$-lattices,
the rationale is to apply its classic version \cite{Lekkerkerker1987}
to the embedded $\mathbb{Z}$-lattice and inspect the independence
of lattice vectors over the ring $\mathbb{Z}\left[\xi\right]$. Based
on Eq. (\ref{eq:zizlRelation}), applying the real Minkowski's second
theorem \cite{Lekkerkerker1987} yields 
\[
\prod_{j=1}^{2n}\lambda_{j}^{2}\leq\gamma_{2n}^{2n}\det\left(\mathbf{B}^{\mathbb{R},\mathbb{Z}\left[\xi\right]}\right)^{2},
\]
 where $\lambda_{j}$ denotes the $j$th successive minimum of lattice
$\Lambda\left(\mathbf{B}^{\mathbb{R},\mathbb{Z}\left[\xi\right]}\right)$.
Substitute Eq. (\ref{eq:volume}) into the above equation, we have
\begin{equation}
\prod_{j=1}^{2n}\lambda_{j}^{2}\leq\gamma_{2n}^{2n}\det\left(\mathbf{B}^{\mathbb{R},\mathbb{Z}\left[i\right]}\right)^{2}\det\left(\Phi^{\mathbb{Z}\left[\xi\right]}\right)^{2n}.\label{eq:realMSP}
\end{equation}
Let the $2n$ successive minima of $\mathcal{L}\left(\mathbf{B}^{\mathbb{R},\mathbb{Z}\left[\xi\right]}\right)$
be $\left\Vert \mathbf{\mathbf{B}^{\mathbb{R},\mathbb{Z}\left[\xi\right]}}\mathbf{x}_{1}\right\Vert ,\ldots\thinspace,\left\Vert \mathbf{\mathbf{B}^{\mathbb{R},\mathbb{Z}\left[\xi\right]}}\mathbf{x}_{2n}\right\Vert .$
W.l.o.g., we assume the input basis $\mathbf{B}$ has full rank, then
so does $\mathbf{\mathbf{B}^{\mathbb{R},\mathbb{Z}\left[\xi\right]}}$.
For any index $j,\thinspace j'$, 
$  \dim\left(\mathrm{span}_{\mathbb{Z}\left[\xi\right]}\left(\sigma^{-1}\left(\mathbf{\mathbf{B}^{\mathbb{R},\mathbb{Z}\left[\xi\right]}}\mathbf{x}_{j},\mathbf{\mathbf{B}^{\mathbb{R},\mathbb{Z}\left[\xi\right]}}\mathbf{x}_{j'}\right)\right)\right)  =\dim\left(\mathrm{span}_{\mathbb{Z}\left[\xi\right]}\left(\sigma^{-1}\left(\mathbf{x}_{j},\mathbf{x}_{j'}\right)\right)\right).$
Since the coefficients of the successive minima satisfy $\dim\left(\mathrm{span}_{\mathbb{R}}\left(\mathbf{x}_{1},\dots\thinspace,\mathbf{x}_{2n}\right)\right)=2n$,
it yields
\[
\dim\left(\mathrm{span}_{\mathbb{Z}\left[\xi\right]}\left(\sigma^{-1}\left(\mathbf{x}_{1},\dots\thinspace,\mathbf{x}_{2n}\right)\right)\right)=n.
\]

We can design an algorithm to partition $\mathbf{x}_{1},\dots\thinspace,\mathbf{x}_{2n}$
into two groups, each with size $n$. Firstly, note that there exists
an index set $S$ with $|S|=n$ such that $\dim\left(\mathrm{span}_{\mathbb{Z}\left[\xi\right]}\left(\sigma^{-1}\left(\mathbf{x}_{S(1)},\dots\thinspace,\mathbf{x}_{S(n)}\right)\right)\right)=n$.
Secondly, starting from $\mathbf{x}_{S(1)},$ we search for one candidate
in $\mathbf{x}_{[2n]\backslash S}$ in each round, noted as $\mathbf{x}_{S(1)}'$
such that $\left(\mathbf{x}_{S(1)}\right)^{\dagger}\mathbf{x}_{S(1)}'\neq0$.
This procedure continues until all $\mathbf{x}_{[2n]}$ have been
partitioned. It follows from Definition \ref{def:smpok} that $\forall j\in 1,\ldots,n$,
\begin{equation}
\left\Vert \sigma^{-1}\left(\mathbf{\mathbf{B}^{\mathbb{R},\mathbb{Z}\left[\xi\right]}}\mathbf{x}_{S(j)}\right)\right\Vert \leq\left\Vert \sigma^{-1}\left(\mathbf{\mathbf{B}^{\mathbb{R},\mathbb{Z}\left[\xi\right]}}\mathbf{x}_{S(j)}'\right)\right\Vert .\label{eq:all_smaller}
\end{equation}
Based on (\ref{eq:all_smaller}), we have $\prod_{j=1}^{n}\lambda_{j,\mathbb{Z}\left[\xi\right]}^{2}\leq\left(\prod_{j=1}^{2n}\lambda_{j}^{2}\right)^{1/2}.$
Plugging this into (\ref{eq:realMSP}), we have 
\[
\prod_{j=1}^{n}\lambda_{j,\mathbb{Z}\left[\xi\right]}^{2}\leq\gamma_{2n}^{n}\left|\det\left(\mathbf{B}^{\mathbb{R},\mathbb{Z}\left[i\right]}\right)\right|\left|\det\left(\Phi^{\mathbb{Z}\left[\xi\right]}\right)\right|^{n}.
\]

\end{IEEEproof}

\subsection{Definition of Algebraic Lattice Reduction}
A lattice has infinitely many bases.
The process of improving the quality of a given basis by some lattice-preserving
transform is generically called lattice reduction. It is well known
that a transformation matrix should be taken from a set of integer
matrices that are invertible in $\mathbb{Z}$ for a real basis, while
such transforms for a complex basis remain poorly understood. 
Denote $\mathrm{GL}_{n}\left(\mathbb{Z}\left[\xi\right]\right)$ as
the set of invertible matrices in the matrix ring $M_{n\times n}\left(\mathbb{Z}\left[\xi\right]\right)$
and call a matrix in $\mathrm{GL}_{n}\left(\mathbb{Z}\left[\xi\right]\right)$
unimodular. To define algebraic lattice reduction, we introduce a lemma to characterize the ``equivalence'' of two lattice bases. 

\begin{lem}
	Two lattice bases $\mathbf{B}$, $\tilde{\mathbf{B}}$ generate the
	same lattice if and only if there exists a matrix $\mathbf{U}\in\mathrm{GL}_{n}\left(\mathbb{Z}\left[\xi\right]\right)$
	such that $\tilde{\mathbf{B}}=\mathbf{B}\mathbf{U}$. \end{lem}
\begin{IEEEproof}
The proof follows the technique in \cite{Frank2014}.
	First, we show that $\mathbf{B},\tilde{\mathbf{B}}$ generate the
	same lattice if $\tilde{\mathbf{B}}=\mathbf{B}\mathbf{U}$ for a unimodular
	matrix $\mathbf{U}$. Let $\Lambda$ be generated by $\mathbf{B}$
	and let $\tilde{\Lambda}$ be generated by $\tilde{\mathbf{B}}$.
	Any element $\tilde{\mathbf{b}}\in\tilde{\Lambda}$ can be written
	as 
	\[
	\tilde{\mathbf{b}}=\tilde{\mathbf{B}}\mathbf{x}=\mathbf{B}\mathbf{U}\mathbf{x}\in\Lambda,
	\]
	for some $\mathbf{x}\in\mathbb{Z}\left[\xi\right]^{n}$, which shows
	that $\tilde{\Lambda}\subseteq\Lambda$ since $\mathbf{U}\mathbf{x}\in\mathbb{Z}\left[\xi\right]^{n}$.
	On the other hand, if $\mathbf{U}$ is invertible, we have $\mathbf{B}=\tilde{\mathbf{B}}\mathbf{U}^{-1}$
	and a similar argument shows that $\Lambda\subseteq\tilde{\Lambda}$.
	Now we show the invertible condition is $\det\left(\mathbf{U}\right)\in\mathbb{Z}\left[\xi\right]^{\times}$.
	Note that if a ring $\mathbb{Z}\left[\xi\right]$ is from a complex
	quadratic field, then it is commutative (a non-commutative example
	is the matrix ring). For any matrix $\mathbf{U}\in\mathbb{Z}\left[\xi\right]^{n\times n}$,
	it follows from Cramer's rule that $\mathbf{U}^{-1}=\left(\det\left(\mathbf{U}\right)\right)^{-1}\mathrm{adj}\left(\mathbf{U}\right)$
	, where $\mathrm{adj}\left(\mathbf{U}\right)\in\mathbb{Z}\left[\xi\right]^{n\times n}$,
	the adjugate of $\mathbf{U}$ is given by $\left[\mathrm{adj}\left(\mathbf{U}\right)\right]_{j',j}=\left(-1\right)^{j'+j}\mathbf{M}_{j,j'}$,
	where $\mathbf{M}_{j,j'}$ is the minor of $\mathbf{U}$ obtained
	by deleting the $j$th row and the $j'$th column of $\mathbf{U}$.
	Clearly, matrix $\mathbf{U}$ is invertible in $\mathbb{Z}\left[\xi\right]^{n\times n}$
	if and only if $\det\left(\mathbf{U}\right)\in\mathbb{Z}\left[\xi\right]^{\times}$,
	such that $\mathbf{U}^{-1}\in\mathbb{Z}\left[\xi\right]^{n\times n}$.
	
	Second, we show that $\tilde{\mathbf{B}}=\mathbf{B}\mathbf{U}$ for
	a unimodular matrix $\mathbf{U}$ if $\mathbf{B},\tilde{\mathbf{B}}$
	generate the same lattice. Based on the ``if'' condition, there
	are some full-rank transforms $\mathbf{U}_{1}$ and $\mathbf{U}_{2}$
	in $\mathbb{Z}\left[\xi\right]^{n\times n}$ such that $\tilde{\mathbf{B}}=\mathbf{B}\mathbf{U}_{1},\thinspace\mathbf{B}=\tilde{\mathbf{B}}\mathbf{U}_{2}$
	and hence $\tilde{\mathbf{B}}=\tilde{\mathbf{B}}\mathbf{U}_{2}\mathbf{U}_{1}$.
	This implies $\mathbf{U}_{2}\mathbf{U}_{1}$ is an identity matrix.
	As the determinant function is distributive, we have $\det\left(\mathbf{U}_{2}\right)\det\left(\mathbf{U}_{1}\right)=1$,
	with $\det\left(\mathbf{U}_{1}\right),\det\left(\mathbf{U}_{2}\right)\in\mathbb{Z}\left[\xi\right]$.
	Thus $\det\left(\mathbf{U}_{1}\right)$ and $\det\left(\mathbf{U}_{2}\right)$
	are a pair of invertible elements in $\mathbb{Z}\left[\xi\right]$,
	and $\mathbf{U}_{1}\in\mathrm{GL}_{n}\left(\mathbb{Z}\left[\xi\right]\right)$,
	$\mathbf{U}_{2}\in\mathrm{GL}_{n}\left(\mathbb{Z}\left[\xi\right]\right)$.
\end{IEEEproof}

The above lemma suggests we can define lattice reduction for
algebraic lattices based on a unimodular transform induced by $\mathrm{GL}_{n}\left(\mathbb{Z}\left[\xi\right]\right)$.
\begin{defn}[Algebraic lattice reduction]
	\label{def:LR}For a given algebraic lattice $\Lambda^{\mathbb{Z}\left[\xi\right]}$
	with basis $\mathbf{B}\in\mathbb{C}^{n\times n}$, find a new basis
	$\tilde{\mathbf{B}}=\mathbf{B}\mathbf{U}$ with shorter basis vectors,
	where $\mathbf{U}\in\mathrm{GL}_{n}\left(\mathbb{Z}\left[\xi\right]\right)$.
\end{defn}

\subsection{Multiple Short Vectors and Independence Over Finite Fields}
Lattice reduction naturally induces unimodular matrices. 
Results in this subsection shows that lattice reduction offers an additional advantage to algebraic lattice network coding for Gaussian multiple-access channel (MAC) \cite{FSK13,Sun2013}, and integer-forcing linear Multiple-Input Multiple-Output (MIMO) detection \cite{Zhan2014IT}. 
  
Consider an algebraic $\mathbb{Z}\left[\xi\right]$-lattice lifted from a linear code over $\mathbb{F}_{p}$. In the decoding process, we define a ring homomorphism \footnote{If $A$ and $B$ are rings, a ring homomorphism \cite{oggier04tutorial} is a map $f:A \rightarrow B$ that satisfies, for all $a,b \in A$, (i) $f(a+b)=f(a)+f(b)$, (ii) $f(a\cdot b)=f(a)\cdot f(b)$,  (iii) $f(1)=1$.}
$f:\thinspace\mathbb{Z}\left[\xi\right]\rightarrow \mathbb{F}_{p}$.
In a practical, non-asymptotic setting, when
lattice coding is considered, the following theorem shows that a unimodular matrix $\mathbf{U}$ can be used to define a network coding matrix that always has full rank over the code space $\mathbb{F}_{p}$. The practical importance of this theorem is further exemplified by an example in Appendix \ref{ADV:UNI}.

\begin{thm} \label{thm-full}
Given a unimodular matrix   $\mathbf{U}\in\mathrm{GL}_{n}(\mathbb{Z}\left[\xi\right])$, the homomorphism of $\mathbf{U}$ in $\mathbb{F}_{p}$ satisfies $\mathrm{rank}\left(f\left(\mathbf{U}\right)\right)=n$.
\end{thm}
\begin{IEEEproof}
	First, the determinant function for measuring ranks defines a mapping
	$\mathrm{GL}_{n}(\mathbb{Z}\left[\xi\right])\rightarrow\mathbb{Z}\left[\xi\right]^{\times}$
	between general linear group over $\mathbb{Z}\left[\xi\right]$ and
	the group of units $\mathbb{Z}\left[\xi\right]^{\times}$. Since it
	respects the multiplication in both groups, the function $\det\left(\cdot\right)$
	defines a group homomorphism. Second, the determinant function respects
	the morphism $f:\mathrm{GL}_{n}(\mathbb{Z}\left[\xi\right])\rightarrow\mathrm{GL}_{n}(\mathbb{F}_{p})$,
	so it yields
	\[
	f\left(\det\left(\mathbf{U}\right)\right)=\det\left(f\left(\mathbf{U}\right)\right).
	\]
	
	\begin{figure}[th]
		\center
		
		\includegraphics[clip,width=2.4in]{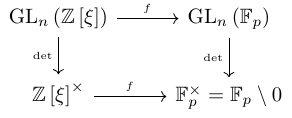}
		
		\protect\caption{The commutative diagram of groups and units. }
		\label{figcmm} 
	\end{figure}
	The composition of morphisms can be represented by the commutative diagram shown in Fig. \ref{figcmm}.
	The non-trivial bottom arrow in the figure holds due to the following reason:
	In a morphism, we have $f(1)=1$ (otherwise we arrive at a contradiction from $f(u \times 1)=f(u) \times f(1)$); thus for a unit $u \in \mathbb{Z}\left[\xi\right]^{\times}$, $f(u \times u^{-1})=f(u) \times f(u^{-1})=1$, which means $f(u)$ has an inverse in $\mathbb{F}_{p}$. This shows the bottom arrow in Fig. \ref{figcmm} holds.
\end{IEEEproof}

\section{\label{alGaussSection}Algebraic Gauss reduction in two dimensions}

Lagrange and Gauss have given the reduction criteria for a two dimensional
real basis. We first generalize this criteria to over complex quadratic
rings. 

\begin{defn} A basis $\mathbf{b}_{1},\mathbf{b}_{2}\in\mathbb{C}^{n}$
	is Gauss reduced if $\|\mathbf{b}_{1}\|\leq\|\mathbf{b}_{2}\|\leq\|\mathbf{b}_{2}+p\mathbf{b}_{1}\|$
	for all $p\in\mathbb{Z}[\xi]$. \end{defn} 

Define the
quantization function $\mathcal{Q}_{\mathbb{Z}\left[\xi\right]}:\mathbb{C}\to\mathbb{Z}[\xi]$ such that
$\mathcal{Q}_{\mathbb{Z}\left[\xi\right]}(x)=\argmin_{\mu\in\mathbb{Z}[\xi]}|x-\mu|$. 
The following algorithm (Algorithm \ref{algeGauss}),
as a special case of algebraic LLL in two dimensions, computes
a Gauss reduced basis.

\begin{algorithm}
	 \KwIn{ basis $\{\mathbf{b}_{1},\mathbf{b}_{2}\}\in\mathbb{C}^{n}$
		of a two dimensional algebraic lattice $\Lambda$, ring
		$\mathbb{Z}[\xi]$ that we want to reduce the basis over.} \KwOut{ reduced basis $\tilde{\mathbf{B}}$.} \BlankLine
	$j=2$\;
	\While{$j\leq 2$}{
		$\mathbf{b}_{2} \leftarrow \mathbf{b}_{2}-\mathcal{Q}_{\mathbb{Z}\left[\xi\right]}\left(\langle\mathbf{b}_{1},\mathbf{b}_{2}\rangle/\|\mathbf{b}_{1}\|^{2}\right)\mathbf{b}_{1}$ \;
      \If{$\|\mathbf{b}_{1}\| > \|\mathbf{b}_{2}\|$}{	swap $\mathbf{b_{1}},\mathbf{b_{2}}$\;}	
      \Else{$j \leftarrow j+1$\;}
      Return $\tilde{\mathbf{B}}=[\mathbf{b}_1,\mathbf{b}_2]$.
}
	\caption{The algebraic Gauss's algorithm.}
	\label{algeGauss}  
\end{algorithm}

\subsection{Performance Characterization}
We begin by introducing the concept of ``fully-reduced''. We say that $x\in\mathbb{C}$ is fully $\mathbb{Z}[\xi]$-reduced
if $|x|\leq|x-q|$ for all $q\in\mathbb{Z}[\xi]$.

\begin{lem} \label{GaussRedlem} Let $x\in\mathbb{C}$ be fully $\mathbb{Z}[\xi]$-reduced.
	Then $|\Re(x)|\leq1/2,|\Im(x)|\leq\sqrt{d}/2$ if $\xi=\sqrt{-d}$
	or $|\Re(x)|\leq1/2$, $|\Im(x)|\leq\frac{1}{\sqrt{d}}\left(-|\Re(x)|+\frac{1+d}{4}\right)$
	if $\xi=\frac{1+\sqrt{-d}}{2}$. \end{lem}

\begin{proof} Define the map $\phi(x+iy)=(x,y)$ for all $x+iy\in\mathbb{C}$.
	Then $|x+iy|=\|(x,y)\|$. When $-d\equiv2,3\mod4$, $\mathbb{Z}[\xi]$
	generates the lattice with basis $(1,0),(0,\sqrt{d})$, otherwise
	$\mathbb{Z}[\xi]$ generates the lattice with basis $(1,0),(1/2,\sqrt{d}/2)$.
	The bounds that form the fundamental region of these lattices correspond
	to the bounds given in the lemma.
\end{proof} 

When the ring of integers
is a Euclidean domain, hereby we prove that Gauss's algorithm returns a basis
corresponding to the successive minima of an algebraic lattice. 

\begin{thm} Let $\tilde{\mathbf{b}}_{1},\tilde{\mathbf{b}}_{2}$
	be an output basis of the algorithm above. Then $\|\tilde{\mathbf{b}}_{1}\|=\lambda_{1},\|\tilde{\mathbf{b}}_{2}\|=\lambda_{2}$
	if $\mathbb{Z}[\xi]$ is the ring of integers of a norm-Euclidean
	domain (i.e., $d=1,2,3,7,11$). \end{thm} 

\begin{proof}
First, the GS coefficients of the output basis, $\mu_{12}=\langle\tilde{\mathbf{b}}_{1},\tilde{\mathbf{b}}_{2}\rangle/\|\tilde{\mathbf{b}}_{1}\|^{2}$ and $\mu_{21}=\langle\tilde{\mathbf{b}}_{1},\tilde{\mathbf{b}}_{2}\rangle/\|\tilde{\mathbf{b}}_{2}\|^{2}$,  are both rounded to zero after the termination of the algorithm.
	 Since no
	swap occurs in the final round, we have $\mathcal{Q}_{\mathbb{Z}\left[\xi\right]}(\mu_{12})=0$ so as to ensure $\|\tilde{\mathbf{b}}_{1}\|\leq\|\tilde{\mathbf{b}}_{2}\|$. Since $\tilde{\mathbf{b}}_{1}$ has been reduced before the final round,  $\mathcal{Q}_{\mathbb{Z}\left[\xi\right]}(\mu_{21})=0$.
	
	To prove that $\|\tilde{\mathbf{b}}_{1}\|=\lambda_{1}$, we denote an arbitrary
	lattice vector $\mathbf{v}=p_{1}\tilde{\mathbf{b}}_{1}+p_{2}\tilde{\mathbf{b}}_{2}$
	where $p_{1},p_{2}\in\mathbb{Z}[\xi]$, and analyze its norm function:
	\begin{equation}
	\|\mathbf{v}\|^{2}=|p_{1}|^{2}\|\tilde{\mathbf{b}}_{1}\|^{2}+|p_{2}|^{2}\|\tilde{\mathbf{b}}_{2}\|^{2}+2\Re(p_{1}^\dagger p_{2}\langle\tilde{\mathbf{b}}_{1},\tilde{\mathbf{b}}_{2}\rangle).\label{prove1_v1}
	\end{equation}
	
	We examine the cases $-d\equiv1,2\mod4$ and $-d\equiv3\mod4$ separately.
	When the chosen ring is in the form of Type I  (i.e., $\xi=\sqrt{-d}$), we let $p_{1}=x+y\sqrt{-d},p_{2}=z+w\sqrt{-d}$
	where $x,y,z,w\in\mathbb{Z}$. Then $p_{1}^\dagger p_{2}=(xz+dyw)+\sqrt{-d}(xw-yz)$,
	and  $	2\Re(p_{1}^\dagger p_{2}\langle\tilde{\mathbf{b}}_{1},\tilde{\mathbf{b}}_{2}\rangle)  =2(xz+dyw)\Re(\langle\tilde{\mathbf{b}}_{1},\tilde{\mathbf{b}}_{2}\rangle)
	-2\sqrt{d}(xw-yz)\Im(\langle\tilde{\mathbf{b}}_{1},\tilde{\mathbf{b}}_{2}\rangle).$
	Since the GS coefficients are fully reduced, we have: 
	\[
	\begin{cases}
	2(xz+dyw)\Re(\langle\tilde{\mathbf{b}}_{1},\tilde{\mathbf{b}}_{2}\rangle)\geq-|xz+dyw|\|\tilde{\mathbf{b}}_{1}\|^{2},\\
	-2\sqrt{d}(xw-yz)\Im(\langle\tilde{\mathbf{b}}_{1},\tilde{\mathbf{b}}_{2}\rangle)\geq-d|xw-yz|\|\tilde{\mathbf{b}}_{1}\|^{2}.
	\end{cases}
	\]
	Based on this, the r.h.s. of Eq. (\ref{prove1_v1}) can be lower bounded:
	\begin{align}
	\|\mathbf{v}\|^{2}\geq Q_{1}'(x,y,z,w)\|\tilde{\mathbf{b}}_{1}\|^{2},\label{eq:lambda_1_case1}
	\end{align}
	where $Q_{1}'(x,y,z,w)$
	\begin{align*}
	 \triangleq(x^{2}+dy^{2}+z^{2}+dw^{2} -|xz+dyw|-d|xw-yz|).
	\end{align*}
	Letting $Q_1(x,y,z,w)\triangleq(x^{2}+dy^{2}+z^{2}+dw^{2} -(xz+dyw)-d(xw-yz))$, we note that the codomain of $Q_1'$ is a subset of the codomain of $Q_1$ (this can be seen by changing the signs of $x,y,z,w$ around until the functions are equivalent), showing positive-definiteness of $Q_1$ immediately yields that $Q_1'$ is also positive-definite. The 4-D symmetric matrix w.r.t. quadratic form $Q_{1}(x,y,z,w)$ can
	be written as 
	\[
	\mathbf{Q}_{1}=\left[\begin{array}{cccc}
	1 & 0 & -\frac{1}{2} & -\frac{d}{2}\\
	0 & d & \frac{d}{2} & -\frac{d}{2}\\
	-\frac{1}{2} & \frac{d}{2} & 1 & 0\\
	-\frac{d}{2} & -\frac{d}{2} & 0 & d
	\end{array}\right].
	\]
	The four eigenvalues of $\mathbf{Q}_{1}$ are: 
	\[
	\begin{cases}
	\frac{d-\sqrt{5d^{2}-6d+9}+3}{4},\\
	\frac{d+\sqrt{5d^{2}-6d+9}+3}{4},\\
	\frac{3d-\sqrt{13d^{2}-6d+1}+1}{4},\\
	\frac{3d+\sqrt{13d^{2}-6d+1}+1}{4}.
	\end{cases}
	\]
	We therefore conclude that $\mathbf{Q}_{1}$ has four positive eigenvalues
	and hence being positive definite with only $d=1,2$ in this case.
	Along with $Q(x,y,z,w)\in\mathbb{Z}$, we arrive at $\|\mathbf{v}\|^{2}\geq\|\tilde{\mathbf{b}}_{1}\|^{2}$
	when $d=1,2$.
	
	When the chosen ring is in the form of Type II (i.e., $\xi=\frac{1+\sqrt{-d}}{2}$), 
	as before, letting $p_{1}=x+y\frac{1+\sqrt{-d}}{2},p_{2}=z+w\frac{1+\sqrt{-d}}{2}$,
	we have $p_{1}^\dagger p_{2}=(xz+1/2(yz+xw)+\frac{1+d}{4}yw)+(\sqrt{-d}/2)(xw-yz)$.
	Then  $	2\Re(p_{1}^\dagger p_{2}\langle\tilde{\mathbf{b}}_{1},\tilde{\mathbf{b}}_{2}\rangle)=  2(xz+1/2(yz+xw) +\frac{1+d}{4}yw)\Re(\langle\tilde{\mathbf{b}}_{1},\tilde{\mathbf{b}}_{2}\rangle) -\sqrt{d}(xw-yz)\Im(\langle\tilde{\mathbf{b}}_{1},\tilde{\mathbf{b}}_{2}\rangle).$
	Using the following inequality from the ``fully-reduced'' constraints:
	$$
	|\Im(x)|\leq\frac{1}{\sqrt{d}}\left(-|\Re(x)|+\frac{1+d}{4}\right),
	$$
	similarly to before, we obtain the inequality
$	\|\mathbf{v}\|^2 \geq (x^2 + xy +\frac{1+d}{4}y^2 + z^2 + zw + \frac{1+d}{4}w^2  - \frac{1+d}{4}|xw-yz|)\|\mathbf{b}_1\|^2 -|\Re(\langle \mathbf{b}_1,\mathbf{b}_2 \rangle)| (|2xz + \frac{1+d}{2}yw + xw +yz| - |xw-yz|).$

	Focusing on the term $(|2xz + \frac{1+d}{2}yw + xw +yz| - |xw-yz|)$, we note that one of the $xw,yz$ on the left hand term must annihilate with one on the right hand term, and one must sum to two times the variable (the choice of which does not matter for our case, as the overall function is symmetric in $xw,yz$). We choose $xw$ to annihilate and $yz$ to coalesce. Then clearly, all terms whose coefficient is $|\Re(\langle \mathbf{b}_1,\mathbf{b}_2 \rangle)|$ are negative, so the minimum is achieved at $|\Re(\langle \mathbf{b}_1,\mathbf{b}_2 \rangle)|=1/2 \|\mathbf{b}_1\|^2$. Now we obtain $\|\mathbf{v}\|^{2}\geq Q_{2}(x,y,z,w)\|\tilde{\mathbf{b}}_{1}\|^{2}$
	with $	Q_{2}(x,y,z,w) \triangleq x^{2}+ xy +\frac{1+d}{4}y^{2}+z^{2}+zw +\frac{1+d}{4}w^{2}-\frac{1+d}{4}xw+\left(\frac{1+d}{4}-1\right)yz  -xz -\frac{1+d}{4}yw.$
	
	The symmetric matrix w.r.t. quadratic form $Q_{2}(x,y,z,w)$ and its
	corresponding eigenvalues are respectively: 
	\[
	\mathbf{Q}_{2}=\left[\begin{array}{cccc}
	1 & 1/2 & -\frac{1}{2} & -\frac{1+d}{8}\\
	1/2 & \frac{1+d}{4} & \frac{1}{2}\left(\frac{1+d}{4}-1\right) & -\frac{1+d}{8}\\
	-\frac{1}{2} & \frac{1}{2}\left(\frac{1+d}{4}-1\right) & 1 & 1/2\\
	-\frac{1+d}{8} & -\frac{1+d}{8} & 1/2 & \frac{1+d}{4}
	\end{array}\right],
	\]
	\[
	\begin{cases}
	\frac{2D+2-\sqrt{9D^{2}-10D^{3}+10-4\frac{D^3-D^2+2}{\sqrt{D^2-2D+2}}}-\sqrt{D^2-2D+2}}{4},\\
	\frac{2D+2+\sqrt{9D^{2}-10D^{3}+10-4\frac{D^3-D^2+2}{\sqrt{D^2-2D+2}}}-\sqrt{D^2-2D+2}}{4},\\
	\frac{2D+2-\sqrt{9D^{2}-10D^{3}+10+4\frac{D^3-D^2+2}{\sqrt{D^2-2D+2}}}+\sqrt{D^2-2D+2}}{4},\\
	\frac{2D+2+\sqrt{9D^{2}-10D^{3}+10+4\frac{D^3-D^2+2}{\sqrt{D^2-2D+2}}}+\sqrt{D^2-2D+2}}{4},
	\end{cases}
	\]
	where $D=\frac{1+d}{4}$. Through checking the eigenvalues, it shows that $\mathbf{Q}_{2}$ is
	positive definite when $d=3,7,11$; therefore $\|\mathbf{v}\|^{2}\geq\|\tilde{\mathbf{b}}_{1}\|^{2}$
	is reached.
	
	To prove that $\|\tilde{\mathbf{b}}_{2}\|=\lambda_{2}$, we leverage the technique
	in \cite{Yao2002}. For both cases of $\xi$, we construct a vector
	$p_{1}\tilde{\mathbf{b}}_{1}+p_{2}\tilde{\mathbf{b}}_{2}$ with $p_{1},p_{2}\in\mathbb{Z}[\xi]$,
	$p_{2}\neq0$. When the chosen ring is in the form of $\xi=\sqrt{-d}$,
	we have $	 \|p_{1}\tilde{\mathbf{b}}_{1}+p_{2}\tilde{\mathbf{b}}_{2}\|^{2} =|p_{2}|^{2}(\|\tilde{\mathbf{b}}_{2}\|^{2}-\|\tilde{\mathbf{b}}_{1}\|^{2})  +(|p_{1}|^{2}+|p_{2}|^{2})\|\tilde{\mathbf{b}}_{1}\|^{2}+2\Re(p_{1}^\dagger p_{2}\langle\tilde{\mathbf{b}}_{1},\tilde{\mathbf{b}}_{2}\rangle)$
	\begin{align*}
	& \geq|p_{2}|^{2}(\|\tilde{\mathbf{b}}_{2}\|^{2}-\|\tilde{\mathbf{b}}_{1}\|^{2})+Q_{1}(x,y,z,w)\|\tilde{\mathbf{b}}_{1}\|^{2}\\
	& \geq(|p_{2}|^{2}-1)(\|\tilde{\mathbf{b}}_{2}\|^{2}-\|\tilde{\mathbf{b}}_{1}\|^{2})+\|\tilde{\mathbf{b}}_{2}\|^{2}\\
	& \geq\|\tilde{\mathbf{b}}_{2}\|^{2}.
	\end{align*}
	This shows $\tilde{\mathbf{b}}_{2}$ is the shortest lattice vector that is
	independent of $\tilde{\mathbf{b}}_{1}$. The proof for the case $\xi=\frac{1+\sqrt{-d}}{2}$
	follows the same way by replacing $Q_{1}(x,y,z,w)$ with $Q_{2}(x,y,z,w)$.
\end{proof}

\subsection{Numerical Examples}
The above
result is further explained through numerical examples. Specifically,
we show how the algorithm finds the two successive minima when the
domain is Euclidean, and how the algorithm fails to work when it is
non-Euclidean.

\textbf{Example 1 (Euclidean domain).} Consider the field $K=\mathbb{Q}\left(\sqrt{-3}\right)$
and its maximal ring of integers $\mathbb{Z}[\omega]$. Suppose the
input lattice basis is 
\[
\mathbf{B}=\left[\begin{array}{cc}
4+\omega & 1+4\omega\\
-1+5\omega & 1+2\omega
\end{array}\right].
\]
The algebraic reduction on this basis will consist of a swap, a size
reduction, and another swap, to yield the reduced basis 
\[
\tilde{\mathbf{B}}=\left[\begin{array}{cc}
-3+3\omega & 1+4\omega\\
2-3\omega & 1+2\omega
\end{array}\right],
\]
which satisfies $\left\Vert \tilde{\mathbf{b}}_{1}\right\Vert ^{2}=\lambda_{1}^{2}=16$,
and $\left\Vert \tilde{\mathbf{b}}_{2}\right\Vert ^{2}=\lambda_{2}^{2}=28$.
On the contrary, if we turn $\mathbf{B}$ into a real basis and perform
real LLL (whose Lovasz's parameter is $1$) on it, the square norm
of the reduced vectors are respectively $16$, $16$, $31$, and $28$.
In its reduced basis, the first two vectors are not independent over
$K$, and the second shortest vector is in the last position. In this
scenario only the Minkowski reduction on the real basis can have the
same effect as our algebraic lattice reduction, whose reduced vectors
respectively have square norms $16$, $16$, $28$, and $28$.

\noindent \textbf{Example 2 (non-Euclidean domain).} Consider the
field $K=\mathbb{Q}(\sqrt{-5})$ and its maximal ring of integers
$\mathbb{Z}[\sqrt{-5}]$. By Definition \ref{def_euclidean}, $\mathbb{Z}[\sqrt{-5}]$ is not Euclidean. We begin with the following basis:
\[
\mathbf{B}=\left[\begin{array}{cc}
2+3\sqrt{-5} & 8+\sqrt{-5}\\
2+\sqrt{-5} & 2
\end{array}\right].
\]
Performing algebraic reduction on this basis consists of a single
size reduction, resulting in the basis 
\[
\tilde{\mathbf{B}}=\left[\begin{array}{cc}
2+3\sqrt{-5} & 6-2\sqrt{-5}\\
2+\sqrt{-5} & -\sqrt{-5}
\end{array}\right].
\]

Such a basis is reduced in the sense of Gauss whose vectors have square
lengths of $58$ and $61$. However, running real LLL over the corresponding
four dimensional basis returns reduced vectors with respective square
lengths $20,30,26,39$. As such, we conclude that the algebraic Gauss's
algorithm does not guarantee an output that corresponds to the successive
minima of the lattice if the chosen ring is not Euclidean.

\section{\label{sec:LLL-over-modules}Algebraic LLL Reduction in High dimensions}

We now introduce the definition of algebraic LLL to address more general higher dimensional lattices.

\begin{defn}[Algebraic LLL]
\label{def:Alll}An $n\times n$ complex matrix $\mathbf{B}\in\mathbb{C}^{n\times n}$
is called an ALLL-reduced basis of lattice $\Lambda^{\mathbb{Z}\left[\xi\right]}\left(\mathbf{B}\right)$
if its QR-decomposition $\mathbf{B}=\mathbf{Q}\mathbf{R}$ satisfies
the following two conditions: 
\begin{equation}
\mathcal{Q}_{\mathbb{Z}\left[\xi\right]}\left(\frac{R_{j,k}}{R_{j,j}}\right)=0,\thinspace\forall\thinspace j<k;\thinspace(\mbox{\ensuremath{\mathrm{size\thinspace reduction\thinspace condition}}})\label{eq:sizecondition}
\end{equation}
\begin{equation}
\delta|R_{j-1,j-1}|^{2}\leq|R_{j,j}|^{2}+|R_{j-1,j}|^{2},(\mbox{\ensuremath{\mathrm{Lov\acute{a}sz's\thinspace condition}}})\label{eq:lovaszcondition}
\end{equation}
$2\leq i\leq n.$ $R_{j,k}$ refers the $\left(j,k\right)$th entry
of $\mathbf{R}$, and $\delta$ is called Lov\'asz's parameter. The feasible range of $\delta$, $\rho_{\mathbb{Z}\left[\xi\right]}^{2}<\delta\leq1$, will be exemplified in this section.
\end{defn}

Based on Lemma \ref{GaussRedlem}, we can make the size reduction condition in (\ref{eq:sizecondition}) explicit, i.e., 
\begin{equation}
	\left|\mathfrak{R}\left(\frac{R_{j,k}}{R_{j,j}}\right)\right|\leq\frac{1}{2}, 	\left|\mathfrak{I}\left(\frac{R_{j,k}}{R_{j,j}}\right)\right|\leq\frac{\sqrt{d}}{2}
\end{equation}
for a Type I ring, and 
\begin{equation}
\left|\mathfrak{R}\left(\frac{R_{j,k}}{R_{j,j}}\right)\right|\leq\frac{1}{2}, 	\left|\Im\left(\frac{R_{j,k}}{R_{j,j}}\right)\right|\leq\frac{1}{\sqrt{d}}\left(-\left|\Re\left(\frac{R_{j,k}}{R_{j,j}}\right)\right|+\frac{1+d}{4}\right)
\end{equation}
for a Type II ring.


We explain how the lower bound of Lov\'asz's parameter $\delta$ should be chosen
based on the covering radius $\rho_{\mathbb{Z}\left[\xi\right]}$
of lattice $\Lambda ^{\mathbb{Z}}\left(\Phi^{\mathbb{Z}\left[\xi\right]}\right)$, where
\[
\rho_{\mathbb{Z}\left[\xi\right]}\triangleq\max_{x\in\mathbb{C}}\left| x-\mathcal{Q}_{\mathbb{Z}\left[\xi\right]}\left(x\right)\right| .
\]
We remind the reader that $\Lambda ^{\mathbb{Z}}\left(\Phi^{\mathbb{Z}\left[\xi\right]}\right)$ denotes a 2-D lattice defined by $\mathbb{Z}\left[\xi\right]$.

\begin{lem}[Covering radius]
\label{thm:CVradiusThm} For an embedded lattice $\Lambda ^{\mathbb{Z}}\left(\Phi^{\mathbb{Z}\left[\xi\right]}\right)$, we have 
\[
\rho_{\mathbb{Z}\left[\xi\right]}=\begin{cases}
\frac{\sqrt{1+d}}{2} & \mathrm{if}\thinspace\xi=\sqrt{-d},\thinspace d>0;\\
\frac{d+1}{4\sqrt{d}} & \mathrm{if}\thinspace\xi=\frac{1+\sqrt{-d}}{2},\thinspace d>0.
\end{cases}
\]
\end{lem}
\begin{IEEEproof}
	Shown in Appendix \ref{ap-prof1}.
\end{IEEEproof}

Since the so-called Siegel's condition \cite{Gama2006} based on rephrasing (\ref{eq:lovaszcondition})
is 
\begin{equation}
\left(\delta-\left|\frac{R_{j-1,j}}{R_{j-1,j-1}}\right|^{2}\right)|R_{j-1,j-1}|^{2}\leq|R_{j,j}|^{2},\label{eq:siegel1}
\end{equation}
it suffices to choose $\delta>\rho_{\mathbb{Z}\left[\xi\right]}^{2}$. 
Commonly used values of $\rho_{\mathbb{Z}\left[\xi\right]}$ are
 shown in Table I.

 \newcolumntype{g}{>{\columncolor{Gray}}c}
\begin{table*}[h!]
	\caption{The covering radiuses of rings in complex quadratic fields (the shaded region satisfies $\rho_{\mathbb{Z}\left[\xi\right]}<1$).}
	\label{tab:mor15-1} \centering 
	\begin{tabular}{|c|g|g|g|g|g|c|c|c|}
		\hline 
		Types of Rings & $d=1$ & $d=2$ & $d=3$ & $d=7$ & $d=11$ & $d=5$ & $d=13$ & $d=15$ \tabularnewline
		\hline 
		\hline 
		$\rho_{\mathbb{Z}\left[\xi\right]}$ & $\frac{\sqrt{2}}{2}$ & $\frac{\sqrt{3}}{2}$ & $\frac{\sqrt{3}}{3}$ & $\frac{2\sqrt{7}}{7}$ & $\frac{3\sqrt{11}}{11}$ & $\frac{\sqrt{6}}{2}$ & $\frac{\sqrt{14}}{2}$ & $\frac{4\sqrt{15}}{15}$ \tabularnewline
		\hline 
		$1-\rho_{\mathbb{Z}\left[\xi\right]}^2$ & $\frac{1}{2}$ & $\frac{1}{4}$ & $\frac{2}{3}$ & $\frac{3}{7}$ & $\frac{2}{11}$ & $-\frac{1}{2}$ & $-\frac{5}{2}$ & $-\frac{1}{15}$ \tabularnewline
		\hline 
	\end{tabular}
\end{table*}

Now we specify the upper bound for $\delta$ and consequently for
$\rho_{\mathbb{Z}\left[\xi\right]}^{2}$ through a potential-function
argument \cite[P. 4790]{Lyu2017}. Define the potential function of
a lattice basis as:
\[
\mathrm{Pot}\left(\mathbf{R}\right)=\prod_{j=1}^{n}\det\left(\Lambda\left(\mathbf{R}_{\Gamma_{i+1}}\right)\right)^{2}=\prod_{j=1}^{n}|R_{j,j}|^{2\left(n-j+1\right)}.
\]
 Let the lattice bases be $\mathbf{R}$ before the swap and $\mathbf{R}'$
after the swap. If Lov\'asz's condition fails to hold, the ratio
of their potential functions is:
\begin{align}
\frac{\mathrm{Pot}\left(\mathbf{R}'\right)}{\mathrm{Pot}\left(\mathbf{R}\right)} & =\frac{\left(|R_{j-1,j}|^{2}+|R_{j,j}|^{2}\right)^{n-j+2}\left(\frac{|R_{j,j}|^{2}|R_{j-1,j-1}|^{2}}{|R_{j-1,j}|^{2}+|R_{j,j}|^{2}}\right)^{n-j+1}}{|R_{j,j}|^{2\left(n-j+1\right)}|R_{j-1,j-1}|^{2\left(n-j+2\right)}} \nonumber\\
 & =\frac{|R_{j-1,j}|^{2}+|R_{j,j}|^{2}}{|R_{j-1,j-1}|^{2}} \nonumber\\
 & <\delta. \label{eq:potdecrease}
\end{align}
Clearly one should ensure $\rho_{\mathbb{Z}\left[\xi\right]}^{2}<\delta\leq1$,
otherwise the algorithm may not converge. By using Lemma \ref{thm:CVradiusThm}
to evaluate the quadratic fields that satisfy $\rho_{\mathbb{Z}\left[\xi\right]}^{2}<1$,
we arrive at the following proposition: 
\begin{prop}
\label{prop:norm_eu}Only the rings from $5$ complex quadratic fields
can be used to define Lov\'asz's condition; they are $\mathbb{Q}\left(\sqrt{-d}\right)$
where $d$ takes the values 
\[
1,2,3,7,11.
\]
\end{prop}

Such rings are all the norm-Euclidean ones in imaginary quadratic fields, because the condition to check the norm-Euclideanity of $\mathbb{Z}\left[\xi\right]$  is exactly  $\rho_{\mathbb{Z}\left[\xi\right]}^{2}<1$ \cite{Kim2017}.

\subsection{Performance Characterisation}

In the following, we set $\delta=\rho_{\mathbb{Z}\left[\xi\right]}^{2}+\epsilon$
with $0<\epsilon\leq1-\rho_{\mathbb{Z}\left[\xi\right]}^{2}$. The
overall performance of algebraic LLL can be described as follows. 
\begin{thm}
\label{prop:odlemma} Let  
$\tilde{\mathbf{B}}=[\tilde{\mathbf{b}}_{1},\ldots\thinspace,\tilde{\mathbf{b}}_{n}]$ be an ALLL-reduced basis w.r.t. an input $\mathbf{B}\in\mathbb{C}^{n\times n}$.   Then   $\tilde{\mathbf{B}}$ admits the following properties:
\begin{align}
 & \left\Vert \tilde{\mathbf{b}}_{1}\right\Vert \leq\epsilon^{-\frac{n-1}{4}} \left|\det\left({\mathbf{B}}\right)\right|^{1/n},\label{eq:lllpr1}\\
 & \left\Vert \tilde{\mathbf{b}}_{1}\right\Vert \leq\epsilon^{-\frac{n-1}{2}}\lambda_{1,\mathbb{Z}\left[\xi\right]},\label{eq:lllpr2}\\
 & \eta_{\mathbb{Z}\left[\xi\right]}(\tilde{\mathbf{B}})\leq\det\left(\Phi^{\mathbb{Z}\left[\xi\right]}\right)^{-n}\prod_{j=1}^{n}\left(1+\rho_{\mathbb{Z}\left[\xi\right]}^{2}\left(\frac{\epsilon^{-1}-\epsilon^{-j}}{1-\epsilon^{-1}}\right)\right)^{1/2}.\label{eq:lllpr3}
\end{align}
\end{thm}
\begin{IEEEproof}
From Siegel's condition (\ref{eq:siegel1}), 
\begin{equation}
|R_{j-1,j-1}|^{2}\leq\epsilon^{-1}|R_{j,j}|^{2}.\label{eq:siegelproof}
\end{equation}
 By induction, it yields 
\[
\left\Vert \tilde{\mathbf{b}}_{1}\right\Vert ^{2}=|R_{1,1}|^{2}\leq\epsilon^{-\left(j-1\right)}|R_{j,j}|^{2},
\]
 $1\leq j\leq n$. Then (\ref{eq:lllpr1}) follows from taking the
product of these inequalities. As for (\ref{eq:lllpr2}), assume that
$x_{1},\ldots,x_{n}\in\mathbb{Z}\left[\xi\right]$ are a set of coprime
numbers such that $\left\Vert \sum_{j=1}^{n}x_{j}\tilde{\mathbf{b}}_{j}\right\Vert =\lambda_{1,\mathbb{Z}\left[\xi\right]}$.
Notice that there must exist one index $k$ with $|x_{k}|\geq1$,
so that $\lambda_{1,\mathbb{Z}\left[\xi\right]}^{2}\geq|x_{k}|^{2}|R_{k,k}|^{2}$.
Then it yields $\lambda_{1,\mathbb{Z}\left[\xi\right]}^{2}\geq|R_{k,k}|^{2}\geq\epsilon^{k-1}\left\Vert \tilde{\mathbf{b}}_{1}\right\Vert ^{2}\geq\epsilon^{n-1}\left\Vert \tilde{\mathbf{b}}_{1}\right\Vert ^{2}$,
which proves (\ref{eq:lllpr2}) . Lastly, in the size reduction condition,
we have $|\frac{R_{j,j'}}{R_{j,j}}|\leq\rho_{\mathbb{Z}\left[\xi\right]}$
$\forall j<j'$, and 
\begin{align}
\left\Vert \mathbf{R}_{1:n,j}\right\Vert ^{2} & =|R_{j,j}|^{2}+\sum_{j'<j}|R_{j',j}|^{2}\nonumber \\
 & \leq|R_{j,j}|^{2}+\sum_{j'<j}\rho_{\mathbb{Z}\left[\xi\right]}^{2}|R_{j',j'}|^{2}\nonumber \\
 & \leq|R_{j,j}|^{2}\left(1+\rho_{\mathbb{Z}\left[\xi\right]}^{2}\left(\epsilon^{-1}+\epsilon^{-2}+\cdots+\epsilon^{-\left(j-1\right)}\right)\right).\label{eq:odstep3}
\end{align}
By substituting the above into the definition, the orthogonality defect
\begin{align*}
\eta_{\mathbb{Z}\left[\xi\right]}(\tilde{\mathbf{B}}) & =\frac{\prod_{j=1}^{n}\left\Vert \mathbf{R}_{1:n,j}\right\Vert }{\det\left(\Phi^{\mathbb{Z}\left[\xi\right]}\right)^{n}\prod_{j=1}^{n}|R_{j,j}|}\\
 & \leq\det\left(\Phi^{\mathbb{Z}\left[\xi\right]}\right)^{-n}\prod_{j=1}^{n}\left(1+\rho_{\mathbb{Z}\left[\xi\right]}^{2}\left(\frac{\epsilon^{-1}-\epsilon^{-j}}{1-\epsilon^{-1}}\right)\right)^{1/2}.
\end{align*}
\end{IEEEproof}

Both (\ref{eq:lllpr1}) and (\ref{eq:lllpr2})
are essentially the same as those of real LLL, while (\ref{eq:lllpr3})
has some factors from ring $\mathbb{Z}\left[\xi\right]$ since its
analysis involves volumes and covering radiuses.
When fixing a common $\delta=1$ for the ALLL over Euclidean rings, we have $\epsilon^{-1}$ equals 
\begin{equation}
2,\,4,\,3/2,\,7/3,\,11/2
\label{epsilonvalues}
\end{equation}
respectively for $d=1,\,2,\,3,\,7,\,11$. \textit{It suggests that ALLL based on Eisenstein integers yields the smallest bound on the shortest vector.}

Of independent interest, we show how far the
size reduction is from the optimal length reduction that employs a
closest vector problem (CVP) algorithm \cite{Lyu2017}, which is useful for the decoding by embedding technique \cite{laura13}.
 
\begin{thm}
\label{prop:odlemma-2}Given a complex basis $\mathbf{B}\in\mathbb{C}^{n\times n}$,
the decoding radius $R_{\mathrm{size}}\triangleq\frac{1}{2}\min_{1\leq j\leq k}|R_{j,j}|$
of size reduction in round $k+1$ satisfies 
\[
R_{\mathrm{size}}\geq\frac{1}{4}\lambda_{1,\mathbb{Z}\left[\xi\right]}V_{2n}^{1/2n}\det\left(\Phi^{\mathbb{Z}\left[\xi\right]}\right)^{-1}\epsilon^{\left(k^{2}-k\right)/4}.
\]
\end{thm}
\begin{IEEEproof}
When LLL is running in the $k+1$th round, then the basis $\left[\tilde{\mathbf{b}}_{1},\ldots\thinspace,\tilde{\mathbf{b}}_{k}\right]$
is LLL-reduced. Let $\left[\tilde{\mathbf{b}}_{1},\ldots\thinspace,\tilde{\mathbf{b}}_{k}\right]=\mathbf{QR}$
denote its QR-decomposition. By using Theorem \ref{thm:MinThms},
we have for $k=2,\ldots,n$, 
\begin{align*}
\lambda_{1,\mathbb{Z}\left[\xi\right]}^{2} & \leq4\left(V_{2k}^{-1/k}\right)\det\left(\Phi^{\mathbb{Z}\left[\xi\right]}\right)^{2} \left|\det\left(\mathcal{L}\left(\left[\tilde{\mathbf{b}}_{1},\ldots\thinspace,\tilde{\mathbf{b}}_{k}\right]\right)\right)\right|^{2/k}\\
 & =4\left(V_{2n}^{-1/n}\right)\det\left(\Phi^{\mathbb{Z}\left[\xi\right]}\right)^{2}\left(\prod_{j=1}^{k}|R_{j,j}|\right)^{2/k}\\
 & \leq4\left(V_{2n}^{-1/n}\right)\det\left(\Phi^{\mathbb{Z}\left[\xi\right]}\right)^{2}\left(\prod_{j=1}^{k}\epsilon^{-\left(k-j\right)/2}|R_{k,k}|\right)^{2/k}\\
 & =4\left(V_{2n}^{-1/n}\right)\det\left(\Phi^{\mathbb{Z}\left[\xi\right]}\right)^{2}|R_{k,k}|^{2}\epsilon^{-\left(k^{2}-k\right)/2}.
\end{align*}
For all quadratic number fields, their packing radiuses are still $1/2$.
By using the definition of the decoding radius, we have 
\begin{align*}
R_{\mathrm{size}} & \triangleq\frac{1}{2}\min_{1\leq j\leq k}|R_{j,j}|\\
 & \geq\frac{1}{4}\lambda_{1,\mathbb{Z}\left[\xi\right]}V_{2n}^{1/2n}\det\left(\Phi^{\mathbb{Z}\left[\xi\right]}\right)^{-1}\min_{1\leq j\leq k}\epsilon^{\left(j^{2}-j\right)/4}.
\end{align*}
\end{IEEEproof}

\subsection{\label{subsec:ImplemenLLL}Implementation and Complexity}

Regarding the implementation of $\mathcal{Q}_{\mathbb{Z}\left[\xi\right]}\left(\cdot\right)$,
for a Type I ring we have
\[
\mathcal{Q}_{\mathbb{Z}\left[\xi\right]}\left(x\right)=\left\lfloor \mathfrak{R}\left(x\right)\right\rceil +i\sqrt{d}\left\lfloor \mathfrak{I}\left(x\right)/\sqrt{d}\right\rceil ,
\]
because its lattice basis $\Phi^{\mathbb{Z}\left[\xi\right]}$ is
orthogonal. 

For a Type II ring, although implementing a sphere decoding algorithm
on basis $\Phi^{\mathbb{Z}\left[\xi\right]}$ suffices, there exist
simpler methods for doing so. For any $\lambda=a+\frac{1+\sqrt{-d}}{2}b,$
$a,b\in\mathbb{Z}$, if $b=2k,\thinspace k\in\mathbb{Z}$, then $\lambda=\left(a+k\right)+\sqrt{-d}k$.
If $b=2k+1,\thinspace k\in\mathbb{Z}$, then $\lambda=\left(a+k\right)+\frac{1}{2}+\sqrt{-d}k+\frac{\sqrt{-d}}{2}$.
Then we can see that $\mathbb{Z}\left[\xi\right]$ is simply the union
of a rectangular lattice $\mathbb{Z}\left[\sqrt{-d}\right]$ and its
coset $\mathbb{Z}\left[\sqrt{-d}\right]+d^{*}$, $d^{*}\triangleq\frac{1}{2}+\frac{\sqrt{-d}}{2}$.
Two examples of such lattices are reproduced in Fig. \ref{fig1_lambda1-1}.
In summary, for a Type II ring we have:
\begin{align*}
 & \mathcal{Q}_{\mathbb{Z}\left[\xi\right]}\left(x\right)=\arg\min_{y}\left| y-x\right| ,\\
 & y\in\left\{ \mathcal{Q}_{\mathbb{Z}\left[\sqrt{-d}\right]}\left(x\right),\mathcal{Q}_{\mathbb{Z}\left[\sqrt{-d}\right]}\left(x-d^{*}\right)+d^{*}\right\} .
\end{align*}

\begin{figure}[th]
\center

\includegraphics[clip,width=0.4\textwidth]{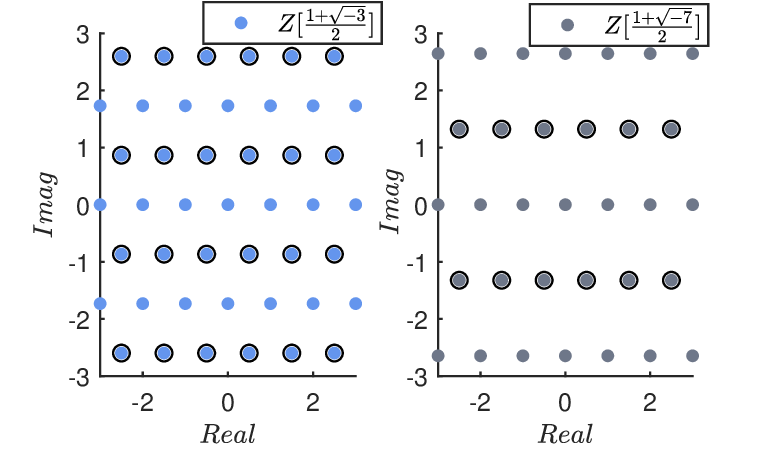}

\protect\caption{Representing a Type II ring by a rectangular ring and its coset. Dots
with open circles represent the cosets.}
\label{fig1_lambda1-1} 
\end{figure}

Now we present the pseudo-code of algebraic LLL in Algorithm
\ref{algeLLL}. Compared with the complex LLL algorithm in \cite{Gan2009}, the
major differences are: i) The rounding function in Step 5 is generalized
from over $\mathbb{Z}\left[i\right]$ to over $\mathbb{Z}\left[\xi\right]$.
ii) Formulas (7)-(15) in \cite{Gan2009} are simplified as a rotation
by quaternions, which is represented by Steps 13-16 of Algorithm \ref{algeLLL}.
The details are given in Appendix \ref{sec:Rotation-and-Quaternions}.

\begin{algorithm}
\KwIn{ lattice basis $\mathbf{B}\in\mathbb{C}^{n\times n}$, Lov\'asz's
parameter $\delta$, ring
$\mathbb{Z}[\xi]$ that we want to reduce the basis over.} \KwOut{reduced
basis $\tilde{\mathbf{B}}\in\mathbb{C}^{n\times n}$, unimodular matrix
$\mathbf{U}$ that makes $\tilde{\mathbf{B}}=\mathbf{B}\mathbf{U}$.}
$[\mathbf{Q},\mathbf{R}]=\mathrm{qr}(\mathbf{B})$; \Comment{The
QR decomposition of $\mathbf{B}$}\; $j=2$, $\mathbf{U}=\mathbf{I}_{n}$\;

\While{$j\leq n$} { 
	\textcolor{blue}{$\mathbf{s}=\mathbf{R}_{1:n,j}$, $\mathbf{t}=\mathbf{U}_{1:n,j}$; \Comment{included for boosted LLL}\;}
	\For{$k=j-1:-1:1$}{ $c=\mathcal{Q}_{\mathbb{Z}\left[\xi\right]}\left(\frac{R_{j,k}}{R_{j,j}}\right)$;
\Comment{Ring quantization}\;
 \If{$c\neq0$}{$\mathbf{R}_{1:n,j}\leftarrow\mathbf{R}_{1:n,j}-c\mathbf{R}_{1:n,k}$\;
$\mathbf{U}_{1:n,j}\leftarrow\mathbf{U}_{1:n,j}-c\mathbf{U}_{1:n,k}$\;}
} 
\textcolor{blue}{\If{$\|\mathbf{s}\|< \|\mathbf{R}_{1:n,j}\|$ and $\mathcal{Q}_{\mathbb{Z}\left[\xi\right]}\left(\frac{s_{j-1}}{R_{j-1,j-1}}\right)==0 $}{$\mathbf{R}_{1:n,j}=\mathbf{s}$, $\mathbf{U}_{1:n,j}=\mathbf{t}$;\Comment{included for boosted LLL} }}

\If{$\delta|R_{j-1,j-1}|^{2}>|R_{j,j}|^{2}+|R_{j-1,j}|^{2}$}
{ define $\mathbf{M}_{v^{*}}\triangleq\left[\begin{array}{cc}
\frac{R_{j-1,j}^{\dagger}}{\sqrt{|R_{j-1,j}|^{2}+|R_{j,j}|^{2}}} & \frac{R_{j,j}^{\dagger}}{\sqrt{|R_{j-1,j}|^{2}+|R_{j,j}|^{2}}}\\
\frac{-R_{j,j}}{\sqrt{|R_{j-1,j}|^{2}+|R_{j,j}|^{2}}} & \frac{R_{j-1,j}}{\sqrt{|R_{j-1,j}|^{2}+|R_{j,j}|^{2}}}
\end{array}\right]$\; swap $\mathbf{R}_{1:n,j}$ and $\mathbf{R}_{1:n,j-1}$, $\mathbf{U}_{1:n,j}$
and $\mathbf{U}_{1:n,j-1}$\; $\mathbf{R}_{j-1:j,1:n}\leftarrow\mathbf{M}_{v^{*}}\mathbf{R}_{j-1:j,1:n}$;
\Comment{Left rotation}\; $\mathbf{Q}_{1:n,j-1:j}\leftarrow\mathbf{Q}_{1:n,j-1:j}\mathbf{M}_{v^{*}}^{-1}$;
\Comment{Right rotation}\; $j\leftarrow\max(j-1,2)$\;} \Else{
$j\leftarrow j+1$\;} } $\tilde{\mathbf{B}}=\mathbf{Q}\mathbf{R}.$ 
\caption{The algebraic LLL algorithm.}
\label{algeLLL}  
\end{algorithm}

Then we analyze the number of loops in the above algorithm. Denote the number of positive and negative tests in Step 12 as $K^{+}$ and $K^{-}$, respectively.
As $K^{-}+K^{+}\leq 2 K^{-} +n -1$ \cite{Lyu2017}, it suffices to bound $K^{-}$. Based on (\ref{eq:potdecrease}), the potential function of the basis decreases in a $\log_{1/\delta}$ scale for each negative tests, and
 we have \cite{Lyu2017,jal08}
\[K^{-}\leq \frac{1}{\ln (1/\delta)} \ln (\kappa(\mathbf{B})^{n(n+1)/2)}),\]
where $\kappa(\mathbf{B})$ denotes the condition number of $\mathbf{B}$. For an $n \times n$ standard complex Gaussian random matrix $\mathbf{B}$, it has been shown \cite{gau05} that 
\[\mathbb{E}(\kappa(\mathbf{B})) < \ln(n)+2.24.\]
The average number of negative tests w.r.t. such input bases is therefore bounded by
\[\mathbb{E}(K^{-}) < \frac{n(n+1)}{2 \ln (1/\delta)} (\ln(n)+2.24).\]
The counterpart of $\mathbb{E}(K^{-})$ is basically the same for real standard Gaussian random matrices (as the input bases of LLL), whose expected condition number is upper bounded by $\ln(n)+2.258$ \cite{gau05}.

The size reduction (Lines 5-9 of Algorithm \ref{algeLLL}) dominates the complexity in each loop, whose complexity is $O(n^2)$. Although this $O(n^2)$ complexity is independent of the chosen ring, the hidden constant of a Type II ring is larger due to its more complicated quantization function.

Finally, let $\delta$ be a factor independent of $n$,
 the overall average complexity of algebraic LLL is $O(n^4 \ln (n))$.

\subsection{Beyond Algebraic LLL}
It is also possible to define the algebraic versions of boosted LLL (shorter basis length) \cite{Lyu2017}, deep LLL (shorter basis vectors)\cite{Schnorr1994} and BKZ (shorter basis vectors)\cite{Chen2011b}. 
Specifically, a simple form of boosted LLL only includes an additional rejection to LLL, and its algebraic version follows in the same vein.
The additional codes for boosted LLL has been marked blue in Algorithm \ref{algeLLL}.
If one intends to design an algebraic BKZ algorithm, the SVP subroutine in BKZ can employ the number of units to speedup the algorithm. 
For instance, only $1/4$ of the points within a Euclidean ball need
to be enumerated in $\mathbb{Z}\left[ i \right]$-lattices as $|\mathbb{Z}\left[i\right]^{\times}|=4$, and only $1/6$ of the points need
to be enumerated in $\mathbb{Z}\left[ \omega \right]$-lattices as $|\mathbb{Z}\left[\omega\right]^{\times}|=6$.
 
\section{\label{sec:app-lr-cf}Numerical Results}
In this section, we numerically verify the efficiency of the proposed algebraic lattice reduction algorithm. The purpose is to demonstrate that algebraic algorithms outperform their non-algebraic real counter-parts, and lattice reduction defined over Eisenstein integers generally yields shorter vectors. To foster reproducible research, MATLAB codes of the algorithms are open source and freely available at GitHub.\footnote{https://github.com/shx-lyu/algebraic-lll}

The types of lattice bases we considered are:

\noindent Type-I: Bases in compute-and-forward  \cite{FSK13,Sun2013,Tunali2015,jerry2018,yuan16}. A target basis $\mathbf{B}$
is decomposed from $\mathbf{M}_{\mathrm{CF}}=\mathbf{B}^\dagger \mathbf{B}$ where $\mathbf{M}_{\mathrm{CF}}= \mathbf{I}_{n} - \frac{P}{P ||\mathbf{h}||^2 +1 }\mathbf{h}\mathbf{h}^{\dagger}$ and $\mathbf{h}\sim\mathcal{CN}(\mathbf{0},\mathbf{I}_n)$. The quality of the bases are controlled by the signal-to-noise
ratio (SNR) parameter $P$.
 
\noindent Type-II: Bases in lattice-reduction-aided and integer-forcing MIMO detection \cite{park11,park12,ahmad13,wkma14,Zhan2014IT,Fischer19}. 
A target basis $\mathbf{B}$
is decomposed from $\mathbf{M}_{\mathrm{IF}}=\mathbf{B}^\dagger \mathbf{B}$ where $\mathbf{M}_{\mathrm{IF}}=\left(\mathbf{H}^\dagger\mathbf{H}+P^{-1}\mathbf{I}_n\right)^{-1}$ and entries of $\mathbf{H}$ are taken from $\mathcal{CN}({0},1)$. 

\noindent Type-III: Bases in the quadratic version of NTRU crytosystem (i.e., GNTRU \cite{GNTRU06} and ETRU \cite{ETRU15}). It considers the $2 n$-dimensional lattice in $\mathbb{Z}[\xi]^{2n}$ spanned by the columns of the basis matrix
\begin{equation}
\mathbf{B}=
\left[\begin{array}{cc}
\mathbf{I}_{n} & \mathbf{0}\\
\mathcal{H} & q\mathbf{I}_{n}
\end{array}\right], \label{eq_ntrubasis}
\end{equation}
where $\mathcal{H}$ is a circulant matrix 
 corresponding to the public key polynomial $\bar{\mathbf{h}}$. E.g., the first column of $\mathcal{H}$ indicated by $\bar{\mathbf{h}}$ is a pseudo-random vector defined over Eisenstein integers in ETRU. Unlike the first two scenarios, here the type of ring has been fixed when given a lattice basis.

\subsection{\label{sub:Real-vs.-Algebraic} Type-I Bases}
In the first example, we consider Type-I bases and the lattices
are respectively defined over Euclidean rings $\mathbb{Z}[i]$, $\mathbb{Z}[\sqrt{-2}]$,
and $\mathbb{Z}[\omega]$, and non-Euclidean rings $\mathbb{Z}[\sqrt{-5}]$,
$\mathbb{Z}[\frac{1+\sqrt{-19}}{2}]$, and $\mathbb{Z}[\frac{1+\sqrt{-39}}{2}]$.
We implement both algebraic LLL reductions and classic LLL reductions, with Lov\'asz's parameter $\delta=0.99$.

\begin{figure}
	%
	%
	%
	
	\subfloat[$n=8$, $P=10\mathrm{dB}$.]{\protect\includegraphics[width=0.4\textwidth]{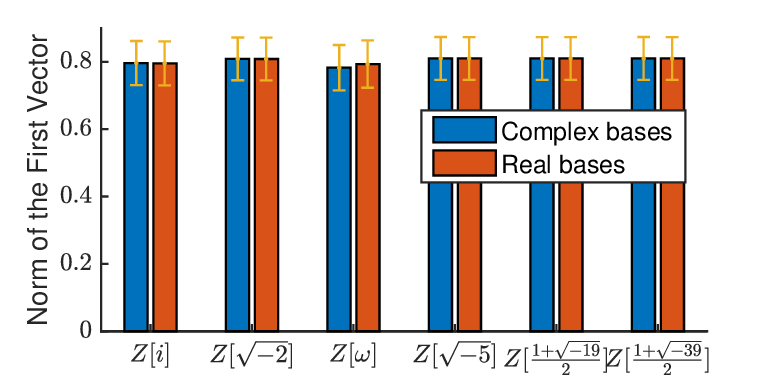}
		
	}
	
	\subfloat[$n=8$, $P=40\mathrm{dB}$.]{\protect\includegraphics[width=0.4\textwidth]{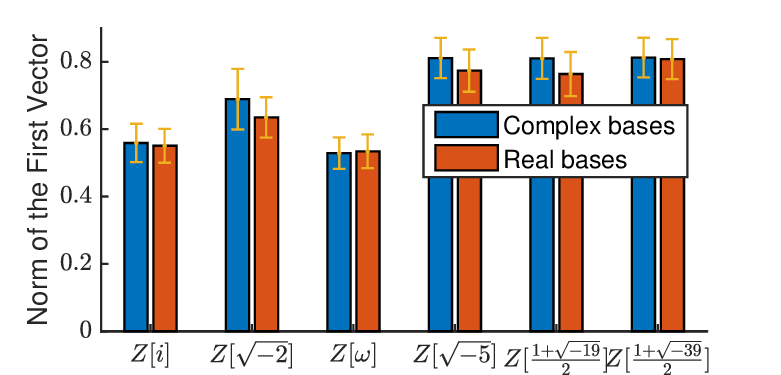}
		
	}
	
	\caption{The Euclidean norm of the first basis vector after different LLL reductions.}
	\label{fig_norm_LLL-1} 
\end{figure}

In Fig. \ref{fig_norm_LLL-1}, we plot the averaged Euclidean norm
of the first basis vector after different reduction approaches. The error bars denote the standard deviations of the objective values, and the legend ``Real bases'' denotes real
lattices generated from expanding $\mathbb{Z}\left[\xi\right]$-based
``Complex bases''. \textit{We can observe from Fig. \ref{fig_norm_LLL-1} that  $\mathbb{Z}[\omega]$-ALLL  finds shorter vector than other algebraic LLL. This can be explained by the fact that $\mathbb{Z}[\omega]$ has the smallest Euclidean minimum, so $\epsilon^{-1}$ in Theorem {\ref{prop:odlemma}} is smaller than others.}\footnote{The values are reflected by (\ref{epsilonvalues}).} Another observation is that, the $\mathbb{Z}[\omega]$-ALLL generates shorter
vectors than its real counter-part, while this is not true for other rings. 

%

\begin{figure}
	%
	%
	%
	
	\subfloat[$n=8$, $P=10\mathrm{dB}$.]{\protect\includegraphics[width=0.4\textwidth]{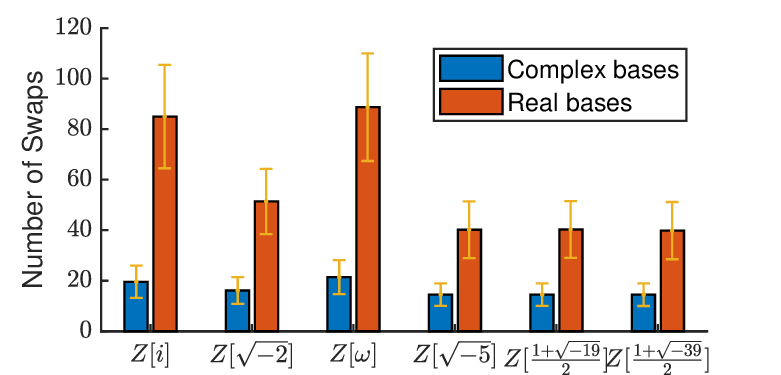}
		
	}
	
	\subfloat[$n=8$, $P=40\mathrm{dB}$.]{\protect\includegraphics[width=0.4\textwidth]{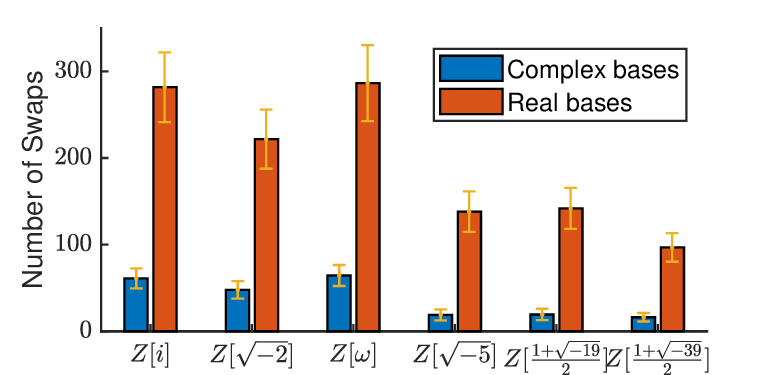}
		
	}
	
	\caption{The complexity of different algebraic LLL algorithms.}
	\label{fig_Complexity_LLL} 
\end{figure}

In Fig. \ref{fig_Complexity_LLL}, we plot the averaged number of
swaps when implementing algebraic/real LLL reduction, as this
metric can reflect the overall complexity of the algorithms. The sub-figures
show that algebraic LLL algorithms have only about $25\%$ complexity
w.r.t. their real counter-parts. This observation is not a surprise as we are dealing with lattices of smaller dimensions. Moreover, the complexity is roughly
inverse-proportional to $\det\left(\Phi^{\mathbb{Z}\left[\xi\right]}\right)$.

\subsection{Type-II bases}

\begin{figure}
	\center
	
	\includegraphics[clip,width=0.4\textwidth]{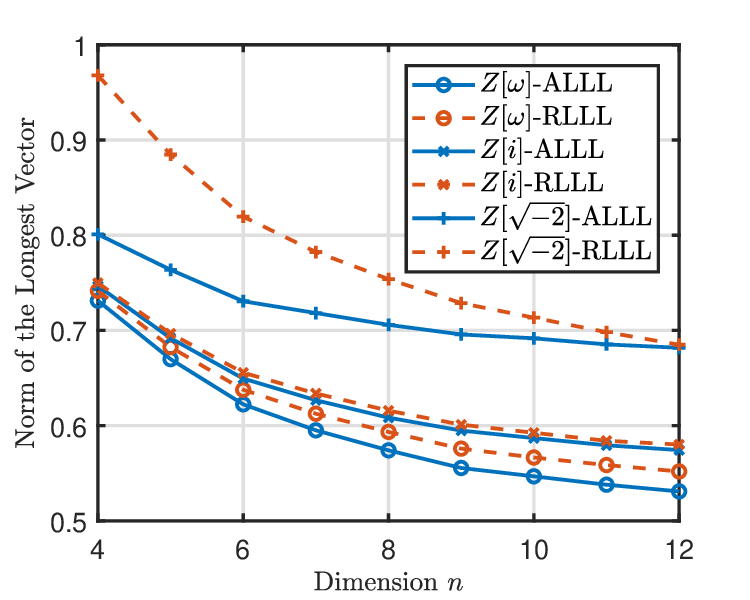}
	
	\protect\caption{The lengths of reduced bases of different algorithms. }
	\label{fig_if_len} 
\end{figure}

\begin{figure}
	\center
	
	\includegraphics[clip,width=0.4\textwidth]{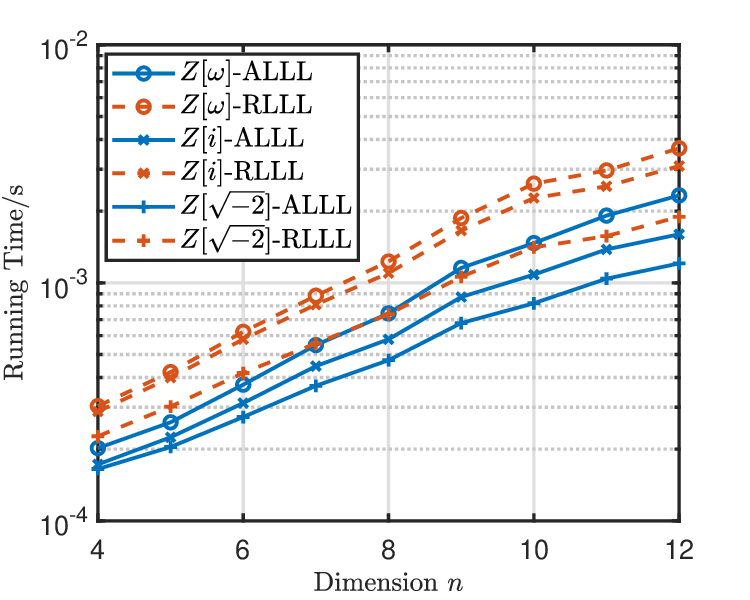}
	
	\protect\caption{The running time of different algorithms. }
	\label{fig_if_time} 
\end{figure}

Since non-Euclidean rings fail to work with most cases, we confine the chosen rings in this example as Eisenstein integers $\mathbb{Z}\left[\omega\right]$, Gaussian integers $\mathbb{Z}\left[i\right]$, and $\mathbb{Z}\left[\sqrt{-2}\right]$. The applications with Type-II bases are involved with the longest basis vector, so we adopt the boosted version of algebraic LLL and its non-algebraic counterpart \cite{Lyu2017}.

Fig. \ref{fig_if_len} shows the lengths of the longest basis vectors of different algorithms. The figure reveals  that $\mathbb{Z}\left[\omega\right]$, $\mathbb{Z}\left[i\right]$, and $\mathbb{Z}\left[\sqrt{-2}\right]$-based algebraic LLL all feature shorter longest vectors than their real counterparts, and $\mathbb{Z}\left[\omega\right]$-ALLL yields the shortest vectors among all these rings. These observations are consistent with those in Type-I bases.

Fig. \ref{fig_if_time} plots the averaged running time of these algorithms. While the algebraic LLL algorithms enjoy a gain of approximately $2$-dimensions (e.g., the time used of $10$-dimensional $\mathbb{Z}\left[\omega\right]$-RLLL suffices to solve a $12$-dimensional $\mathbb{Z}\left[\omega\right]$-ALLL), more compact rings in general cost more time. The non-algebraic $\mathbb{Z}\left[\omega\right]$-RLLL consumes the largest amount of time.

\subsection{Type-III bases}
In the last example, we examine the advantage of using algebraic LLL to reduce/pre-process bases defined by GNTRU \cite{GNTRU06} and ETRU \cite{ETRU15} cryptosystems. The $\mathcal{H}$ factor in Eq. (\ref{eq_ntrubasis}) is defined over Eisenstein integers $\mathbb{Z}\left[\omega\right]$ for ETRU and over Gaussian integers $\mathbb{Z}\left[i\right]$ for GNTRU. 
Since the algebraic and non-algebraic algorithms have almost \textit{identical} performance in the length-metric, we focus on showing the time-advantage of algebraic LLL. We set $q=383$ according to  \cite{ETRU15}.

In Fig. \ref{fig_ntru_time}, we plot the ratio of running time of algebraic over non-algebraic algorithms. As the dimension $2n$ rises from $4$ to $28$, the $\mathbb{Z}\left[\omega\right]$-based LLL only consumes about $50\%$ of RLLL's running time, and the $\mathbb{Z}\left[i\right]$-based LLL only consumes about $35\%$ of RLLL's running time.

Note that our lattice reduction algorithm over quadratic fields can also cryptanalyze a general NTRU crypto system \cite{HoffsteinPS98,StehleS11} defined over cyclotomic fields, as the Kronecker-Webers Theorem \cite{Kronecker-Weber1974} guarantees that quadratic fields are subfields of cyclotomic fields. A fully characterized example is shown in Appendix \ref{sec:cryptNTRU}.

\begin{figure}
	\center
	
	\includegraphics[clip,width=0.4\textwidth]{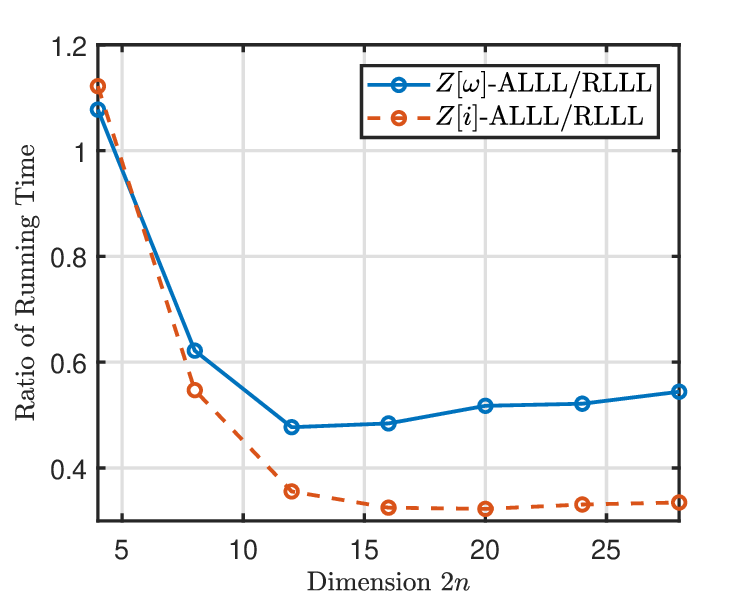}
	
	\protect\caption{The ratio of running time of algebraic over non-algebraic algorithms. }
	\label{fig_ntru_time} 
\end{figure}

\section{Conclusions}
In this work, we have investigated the properties of algebraic lattices and
 the proper design of Gauss and LLL reduction operating in the
algebraic domain. We have shown that, within Euclidean domains, it is possible to successfully build
an algebraic Gauss's algorithm that returns a basis that corresponds to
the successive minima for a two dimensional basis. Moreover, the convergence of algebraic LLL also requires the ring to be Euclidean. Our simulations results show that algebraic algorithms can not only run faster than non-algebraic algorithms, for Type-II bases they also generate shorter vectors.

\appendices{}

\section{\label{ADV:UNI}The advantage of using unimodular matrices}
In the non-asymptotic setting, without loss of generality, choose an error code over $\mathbb{F}_5$, which is further employed to build a lattice code over Gaussian integers. Specifically, elements in $\mathbb{F}_5$ are uniquely mapped to the coset leaders of $\mathbb{Z}[i]/(2+i)\mathbb{Z}[i]$, i.e.,
\[\{0,1,2,3,4\}\rightarrow \{3+4i,1,2+2i,-1,3+3i\}.\] 
Its inverse map from $\mathbb{Z}[i]$ to $\mathbb{F}_5$ is referred to as a homomorphism $f$.

Consider the application to integer-forcing \cite{Zhan2014IT} with $2$ transmit antennas and $2$ receive antennas. With dedicated algorithms, the receiver computes a network coding matrix $\mathbf{A}\in \mathbb{Z}[i]^{2 \times 2}$,  its linear combinations of messages $\mathbf{W}'$, and aims to inverse the following equation over $\mathbb{F}_5$:
\[\mathbf{W}'=f(\mathbf{A})\cdot \mathbf{W},\]
where $\mathbf{W}$ denotes the desired message matrix.

Based on a non-lattice reduction algorithm, assume that $\mathbf{A}$ has the form of
\[\mathbf{A}= \left[\begin{array}{cc}
3+2i & -1+i\\
3+4i & -5
\end{array}\right].\]
Then we obtain 
\[f(\mathbf{A})= \left[\begin{array}{cc}
2 & 4\\
0 & 0
\end{array}\right],\]
and $f(\mathbf{A})$ has no matrix inverse. On the contrary, if we employ a lattice reduction algorithm to design $\mathbf{A}$, then $\mathbf{A}$ must be a unimodular matrix, and Theorem \ref{thm-full} guarantees that $f(\mathbf{A})$ has full rank. Moreover, the matrix inverse is very simple based on unimodular matrices: $f(\mathbf{A})^{-1}=f(\mathbf{A}^{-1})$, in which the first inverse is over $\mathbb{F}_5$, and the second is over $\mathbb{Z}[i]$.

\section{\label{prof:lem1}Proof of Lemma \ref{thm:CVradiusThm} }\label{ap-prof1}
The covering radius $\rho_{\mathbb{Z}\left[\xi\right]}$
 can be analyzed through describing the relevant vectors of the Voronoi region of $\Lambda ^{\mathbb{Z}}\left(\Phi^{\mathbb{Z}\left[\xi\right]}\right)$. For a real lattice $\Lambda^{\mathbb{Z}}\left(\mathbf{B}^{\mathbb{R}}\right)$,
$\mathbf{B}^{\mathbb{R}}\in\mathbb{R}^{n\times n}$, its Voronoi region around the origin is
\[
\mathcal{V} \triangleq	\left\{\mathbf{x}\in \mathbb{R}^{n} \mathrel{\Big|} \left\Vert \mathbf{x} \right\Vert 
\leq
\left\Vert \mathbf{x}-\mathbf{t} \right\Vert \, \forall \mathbf{t}\in \Lambda^{\mathbb{Z}}\left(\mathbf{B}^{\mathbb{R}}\right), \mathbf{t}\neq \mathbf{0} 
\right\} .
\]
The points $\mathbf{p}$ of the lattice for which the hyper-plane between
$\mathbf{0}$ and $\mathbf{p}$ contains a facet of $\mathcal{V}$ are called the Voronoi relevant
vectors.
\begin{IEEEproof}
	Any given lattice $\Lambda^{\mathbb{Z}}\left(\mathbf{B}^{\mathbb{R}}\right)$ 
	can be partitioned into exactly $2^{n}$ cosets of the form $C_{\mathbf{B}^{\mathbb{R}},\mathbf{p}}=2\Lambda^{\mathbb{Z}}+\mathbf{B}^{\mathbb{R}}\mathbf{p}$
	with $\mathbf{p}\in\left\{ 0,1\right\} ^{n}$. If $\mathbf{s}_{\mathbf{p}}$
	is a shortest vector for $C_{\mathbf{B}^{\mathbb{R}},\mathbf{p}}$,
	then the set $\cup_{\mathbf{p}\in\left\{ 0,1\right\} ^{n}\backslash\left\{ \mathbf{0}\right\} }\left\{ \pm\mathbf{s}_{\mathbf{p}}\right\} $
	contains all the relevant vectors \cite{Viterbo1996}. For embedded lattices of $\mathbb{Z}\left[\xi\right]$
	in $\mathbb{R}^{2}$, their generator matrices must have the forms
	as shown in Eq. (\ref{eq:geMatrixPhi}). We discuss the two scenarios
	separately:
	
	i) If $\xi=\sqrt{-d}$ and $\mathbf{B}^{\mathbb{R}}=\Phi^{\mathbb{Z}\left[\xi\right]}=\left[\begin{array}{cc}
	1 & 0\\
	0 & \sqrt{d}
	\end{array}\right]$, we have 
	\[
	\cup_{\mathbf{p}\in\left\{ 0,1\right\} ^{2}\backslash\left\{ \mathbf{0}\right\} }\left\{ \pm\mathbf{s}_{\mathbf{p}}\right\} =\left\{ \pm\left[1,0\right]^{\top},\pm\left[0,\sqrt{d}\right]^{\top}\right\} .
	\]
	Then the covering radius in this case is $\rho_{\mathbb{Z}\left[\xi\right]}=\frac{\sqrt{1+d}}{2}$.
	
	ii) If $\xi=\frac{1+\sqrt{-d}}{2}$ and $\mathbf{B}^{\mathbb{R}}=\Phi^{\mathbb{Z}\left[\xi\right]}=\left[\begin{array}{cc}
	1 & \frac{1}{2}\\
	0 & \frac{\sqrt{d}}{2}
	\end{array}\right]$, the three cosets with non-zero shifts are 
	\[
	\begin{cases}
	C_{\mathbf{B}^{\mathbb{R}},[1,0]^{\top}}=2\Lambda^{\mathbb{Z}}+\left[1,0\right]^{\top},\\
	C_{\mathbf{B}^{\mathbb{R}},[0,1]^{\top}}=2\Lambda^{\mathbb{Z}}+\left[\frac{1}{2},\frac{\sqrt{d}}{2}\right]^{\top},\\
	C_{\mathbf{B}^{\mathbb{R}},[1,1]^{\top}}=2\Lambda^{\mathbb{Z}}+\left[\frac{3}{2},\frac{\sqrt{d}}{2}\right]^{\top}.
	\end{cases}
	\]
	It follows that 
	\begin{align*}
	& \cup_{\mathbf{p}\in\left\{ 0,1\right\} ^{2}\backslash\left\{ \mathbf{0}\right\} }\left\{ \pm\mathbf{s}_{\mathbf{p}}\right\} \\
	& =\left\{ \pm\left[1,0\right]^{\top},\pm\left[\frac{1}{2},\frac{\sqrt{d}}{2}\right]^{\top},\pm\left[\frac{1}{2},-\frac{\sqrt{d}}{2}\right]^{\top}\right\} .
	\end{align*}
	So the point in $\mathcal{V}$ that has the maximum distance to the origin can be obtained as the
	intersection between line $y=-\frac{1}{\sqrt{d}}x+\frac{d+1}{4\sqrt{d}}$
	and line $x=0$ (or line $x=\frac{1}{2}$). Lastly we obtain $\rho_{\mathbb{Z}\left[\xi\right]}=\frac{d+1}{4\sqrt{d}}.$ 
\end{IEEEproof}

\section{\label{sec:Rotation-and-Quaternions}Rotations and Quaternions}

By introducing the concept of quaternions, representations of rotations
become more compact, and unit normalisation of floating point quaternions suffers
from less rounding defects \cite{Dam98quaternions}. Now we explain why quaternions are involved.
As in \cite{Lyu2017}, the pseudo-codes of an LLL algorithm consist
of ``swaps'' and ``size reductions''. After a swap, the structure
of the $\mathbf{R}$ matrix has been destroyed. Since implementing
another factorisation costs too much complexity, we show that the
$\mathbf{R}$ matrix structure can be recovered by left multiplying
the matrix form of a quaternion. With a slight abuse of notation,
let $\left\{ 1,i,j,k\right\} $ be a basis for a vector space of dimension
$4$ over $\mathbb{R}$. These elements satisfy the rules $i^{2}=-1$,
$j^{2}=-1,$$k^{2}=-1$, and $k=ij=-ji$. The Hamilton's quaternions
is a set $\mathbb{H}$ defined by 
\[
\mathbb{H}\triangleq\left\{ x+yi+zj+wk\mid x,y,z,w\in\mathbb{R}\right\} .
\]
For any Hamilton's quaternion $q=x+yi+zj+wk$, it can be written as
\[
\left(x+yi\right)+\left(zj-wji\right)=\alpha_{q}+j\beta_{q},
\]
$\alpha_{q}\in\mathbb{C},$ $\beta_{q}\in\mathbb{C}$. Then $\mathbb{H}$
is also a $\mathbb{C}$-vector space with basis $\left\{ 1,j\right\} $.
Let $\psi\left(q\right)=\left[\alpha_{q},\beta_{q}\right]^{\top}$,
since the multiplication of $q$ with $v=\alpha_{v}+j\beta_{v}$ can
be identified as 
\[
\psi\left(vq\right)=\underset{\triangleq\mathbf{M}_{v}}{\underbrace{\left[\begin{array}{cc}
\alpha_{v} & -\beta_{v}^{\dagger}\\
\beta_{v} & \alpha_{v}^{\dagger}
\end{array}\right]}}\psi\left(q\right),
\]
we call $\mathbf{M}_{v}$ the matrix form of a quaternion $v$.

In the QR-decomposition, $\mathbf{Q}$ denotes a unit in the matrix
ring $M_{n\times n}\left(\mathbb{C}\right)$ since $\det\left(\mathbf{Q}\right)\in\left\{ \pm1,\pm i\right\} $.
Suppose we have in the $t$th round and after a swap that $\mathbf{Q}^{t}\in\mathbb{C}^{2\times2}$, $\mathbf{R}^{t}\in\mathbb{C}^{2\times2}$, then the rotation operation by a quaternion $v^{*}$
is denoted by: 
\[
\mathbf{Q}^{t}\mathbf{R}^{t}=\underset{\triangleq\mathbf{Q}^{t+1}}{\underbrace{\mathbf{Q}^{t}\mathbf{M}_{v^{*}}^{-1}}}\underset{\triangleq\mathbf{R}^{t+1}}{\underbrace{\mathbf{M}_{v^{*}}\mathbf{R}^{t}}}.
\]
Since $\mathbf{Q}^{t}\in M_{2\times2}\left(\mathbb{C}\right)^{\times}$,
$\mathbf{M}_{v^{*}}^{-1}\in M_{2\times2}\left(\mathbb{C}\right)^{\times}$,
we have $\mathbf{Q}^{t+1}\in M_{2\times2}\left(\mathbb{C}\right)^{\times}$.
Denote the first column of $\mathbf{R}^{t}$ as $\left[R_{j-1,j},R_{j,j}\right]^{\top}.$
The rotation is about nulling the second entry, so we can choose the
quaternion as 
\[
v^{*}=\frac{R_{j-1,j}^{\dagger}}{\sqrt{|R_{j-1,j}|^{2}+|R_{j,j}|^{2}}}+j\frac{-R_{j,j}}{\sqrt{|R_{j-1,j}|^{2}+|R_{j,j}|^{2}}}.
\]

\section{\label{sec:cryptNTRU}Cryptanalysis on NTRU}

Consider a cyclotomic field $\mathbb{Q}(\zeta_m)$ defined by a cyclotomic polynomial $\Phi_m(x)$ of degree $\phi(m)$, where $\zeta_m$ denotes the $m$th  root of unity, and $\phi$ denotes Euler's totient function.
To crack the private key based on the given
public key $h\in \mathcal{R} \triangleq \mathbb{Z}[x]/\Phi_m(x)$ in the celebrated NTRU system, we need to find short vectors in a module-lattice $\Lambda(\mathcal{B})$ \cite{HoffsteinPS98,StehleS11,AlbrechtBD16}, whose $2\times 2$ basis is defined by
\[\mathcal{B}= \left[\begin{array}{cc}
q & h\\
0 & 1
\end{array}\right],\]
where $q\in \mathbb{N}$ is a natural number and $h\in \mathcal{R}$.

It is known that we can transform the $2\times 2$ basis $\mathcal{B}$ into a $2\phi(m)\times 2\phi(m)$ basis in $\mathbb{Q}$ and apply the conventional lattice reduction algorithm over real numbers. On the contrary, since 
 quadratic fields are subfields of cyclotomic fields owning to the Kronecker-Webers Theorem \cite{Kronecker-Weber1974}, we can leverage the following Galois extension
 \begin{figure}[h]
 	\center
 	\includegraphics[clip,height=0.25\textwidth]{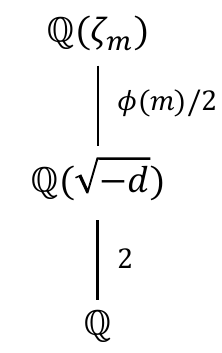}
 	\label{fig_ntru_expan} 
 \end{figure}

\noindent The reduction of basis can therefore be applied to smaller 
basis of dimension $\phi(m)\times \phi(m)$ in $\mathbb{Q}(\sqrt{-d})$, rather than of dimension $2\phi(m)\times 2\phi(m)$ in $\mathbb{Q}$, so as to enjoy shorter running time and sometimes better basis quality.

Consider the instance of $\mathbb{Q}(\zeta_m)=\mathbb{Q}(\zeta_{24})$, $\mathbb{Q}(\sqrt{-d})=\mathbb{Q}(\sqrt{-3})=\mathbb{Q}(\zeta_{3})$, $q=23$, and $h=5+7\zeta_{24}$. We observe that  $\mathbb{Q}(\zeta_{24})/\mathbb{Q}$ is generated by $x^8 - x^4 + 1$, 
 $\mathbb{Q}(\sqrt{-3})/\mathbb{Q}$ is generated by $x^2+x+1$,  and 
  $\mathbb{Q}(\zeta_{24})/\mathbb{Q}(\sqrt{-3})$ is generated by $x^4+\zeta_{3}^2$.
  For lattice reduction over $\mathbb{Q}$, we can adopt the integral basis of $\mathbb{Z}[\zeta_{24}]$:
  \[
  1, \zeta_{24}, ..., \zeta_{24}^7.
  \]
  For lattice reduction over $\mathbb{Q}(\zeta_{3})$, the basis for the $\mathbb{Z}[\zeta_{3}]$-module of rank $4$ is given by
  \[
  1, \zeta_{24}, \zeta_{24}^2, \zeta_{24}^3.
  \]
  The input module-lattice $\Lambda(\mathcal{B})$ 
  features basis vectors of (square) lengths $2116$ and $300$ in both $\mathbb{Q}$ and $\mathbb{Q}(\zeta_{3})$.
  
  By respectively performing LLL over $\mathbb{Q}$ and the proposed quadratic LLL over $\mathbb{Q}(\zeta_{3})$, the output vectors have (square) lengths of
  \begin{align*}
  	&160,
  	160,
  	160,
  	140,
  	192,
  	156,
  	140,
  	196,\\
  	& 160,
  	140,
  	140,
  	160,
  	160,
  	156,
  	192,
  	140, \,(\mathrm{LLL})
  \end{align*}
  and 
    \[140,
    140,
    160,
    140,
    140,
    160,
    160,
    160. \,(\mathrm{algebraic\,LLL})\] 
 This instance shows that algebraic LLL can find shorter vectors than its non-algebraic counterpart while reducing a smaller dimensional basis.

\bibliographystyle{IEEEtranMine}
\bibliography{lib}

\begin{IEEEbiography}[{\includegraphics[width=1in,height=1.25in,clip,keepaspectratio]{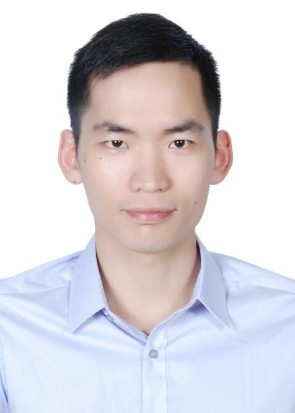}}]{Shanxiang Lyu}
received the B.S. and M.S. degrees in electronic and information engineering from
South China University of Technology, Guangzhou, China, in 2011
and 2014, respectively, and the Ph.D. degree from the
Electrical and Electronic Engineering Department, Imperial College London, UK,
in 2018. 
He is currently a lecturer   
with the College of Cyber Security, Jinan University. He received the superstar supervisor award of the National Crypto-Math Challenge of China in 2020.
His
research interests are in lattice theory, algebraic number theory, and their applications.
\end{IEEEbiography}

\begin{IEEEbiography}[{\includegraphics[width=1in,height=1.25in,clip,keepaspectratio]{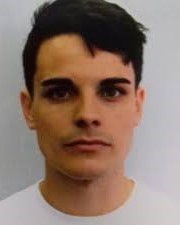}}]{Christian Porter}
	received the M.S. degree in Mathematics from University of York, UK, in 2017. Currently, he is working toward the Ph.D. degree in the
	Electrical and Electronic Engineering Department, Imperial College London, UK. His research focus is primarily on lattice reduction theory for cryptological and coding purposes.
\end{IEEEbiography}

\begin{IEEEbiography}[{\includegraphics[width=1in,height=1.25in,clip,keepaspectratio]{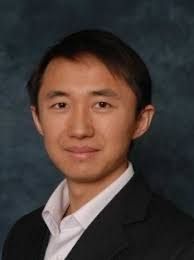}}]{Cong Ling} (S'99-A'01-M'04) 
	received the B.S. and M.S. degrees in electrical engineering from
	the Nanjing Institute of Communications Engineering, Nanjing, China, in 1995
	and 1997, respectively, and the Ph.D. degree in electrical engineering from
	the Nanyang Technological University, Singapore, in 2005.
	He had been on the faculties of the Nanjing Institute of Communications
	Engineering and King's College. He is currently a Reader (Associate Professor) with the Electrical and Electronic Engineering Department, Imperial
	College London. His research interests are coding, information theory, and
	security, with a focus on lattices.
	Dr. Ling has served as an Associate Editor for the IEEE TRANSACTIONS
	ON COMMUNICATIONS and the IEEE TRANSACTIONS ON VEHICULAR
	TECHNOLOGY.\end{IEEEbiography}

\end{document}